\documentclass[9pt,longpaper,twoside,web]{ieeecolor}
\usepackage{gen_settings}
\usepackage{generic}
\usepackage{textcomp}
\def\BibTeX{{\rm B\kern-.05em{\sc i\kern-.025em b}\kern-.08em
    T\kern-.1667em\lower.7ex\hbox{E}\kern-.125emX}}
\begin{document}
\title{Barrier-Based Test Synthesis for Safety-Critical Systems Subject to Timed Reach-Avoid Specifications}
\author{Prithvi Akella, \IEEEmembership{Student Member, IEEE}, Mohamadreza Ahmadi, \IEEEmembership{Member, IEEE}, Richard M. Murray, \IEEEmembership{Fellow, IEEE}, and Aaron D. Ames, \IEEEmembership{Fellow, IEEE}
\thanks{P. Akella, R. M. Murray, and A. D. Ames are with Control and Dynamical Systems (CDS) at the California Institute of Technology, 1200 E. California Blvd., MC 104-44, Pasadena, CA 91125,  e-mail: (\texttt{\{pakella,murray,ames\}@caltech.edu}). M. Ahmadi is with TuSimple, 9191 Towne Centre Dr STE 600, San Diego, CA 92122, e-mail: (\texttt{mohamadreza.ahmadi@tusimple.ai}).  
}}

\maketitle

\begin{abstract}
We propose an adversarial, time-varying test-synthesis procedure for safety-critical systems \textit{without requiring specific knowledge of the underlying controller steering the system.}  From a broader test and evaluation context, determination of difficult tests of system behavior is important as these tests would elucidate problematic system phenomena before these mistakes can engender problematic outcomes, \textit{e.g.} loss of human life in autonomous cars, costly failures for airplane systems, \textit{etc}.  Our approach builds on existing, simulation-based work in the test and evaluation literature by offering a controller-agnostic test-synthesis procedure that provides a series of benchmark tests with which to determine controller reliability.  To achieve this, our approach codifies the system objective as a timed reach-avoid specification.  Then, by coupling control barrier functions with this class of specifications, we construct an instantaneous difficulty metric whose minimizer corresponds to the most difficult test at that system state.  We use this instantaneous difficulty metric in a game-theoretic fashion, to produce an adversarial, time-varying test-synthesis procedure that does not require specific knowledge of the system's controller, but can still provably identify realizable and maximally difficult tests of system behavior.  Finally, we develop this test-synthesis procedure for both continuous and discrete-time systems and showcase our test-synthesis procedure on simulated and hardware examples.
\end{abstract}


\section{Introduction}
\label{sec:introduction}
For safety-critical autonomous systems where failure can mean loss of human life, \textit{e.g.} autonomous cars, autonomous flight vehicles, \textit{etc}, it is natural to ask the following question: does the system's controller effectively render the system safe while steering it in satisfaction of its objective?  Prior work in the literature approaches this question in two distinct ways.  The first line of work aims to develop methods to verify whether a given controller ensures that its system satisfies its desired objective, despite any perturbation from an allowable set.  Termed \emph{verification}, the development and application of these techniques are of widespread study, but not the focus of this paper \cite{clarke2018handbook, baier2008principles,katoen2016probabilistic,platzer2008keymaera,fitting2012first,schumann2001automated}.  The second version of this problem arises primarily in the test and evaluation of safety-critical cyber-physical systems.  Here, either a high-dimensional state space or a requirement to search over system trajectories frustrates the immediate application of the aforementioned verification techniques.  As such, there is a non-trivial amount of work aimed at systematically generating more difficult tests of the onboard control architecture and evaluating system performance with rising test difficulty.  Additionally, the vast majority of these techniques require \textit{a-priori} knowledge of the controller-to-be-tested and a simulator of the system-controller pair.  In a related and, to the best of our knowledge, a less explored vein, we believe the search for problematic system behavior independent of apriori controller knowledge is dually useful, as will be further explained.

\begin{figure}
    \centering
    \includegraphics[width = 261.0 pt]{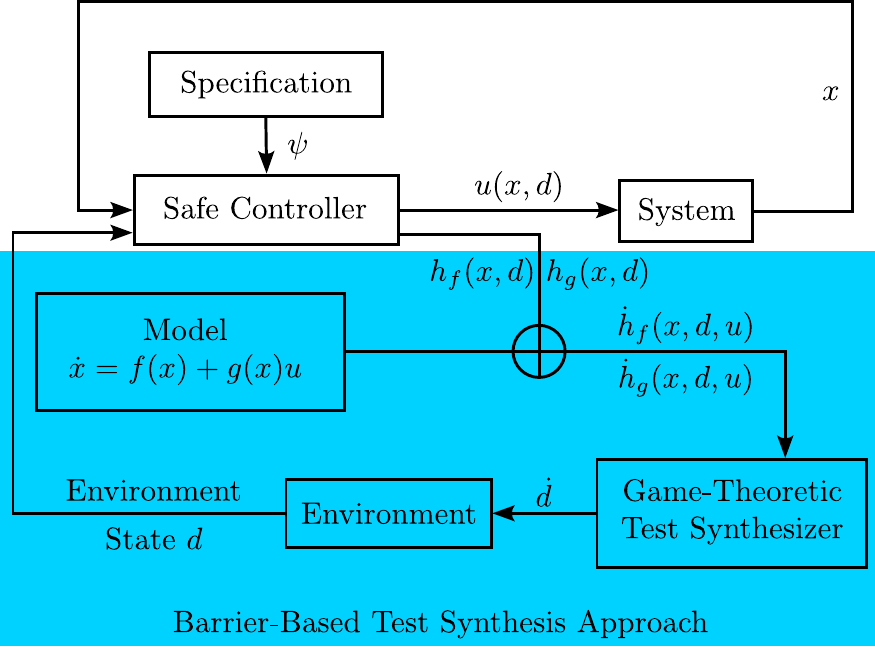}
    \caption{A general flowchart of our test-synthesis procedure for safety-critical systems subject to Timed-Reach Avoid Specifications.  We assume the specification $\psi$ that influences the safe controller is express-able via barrier functions $h_f,h_g$.  Simultaneously, these same barrier functions are used in a game-theoretic test-synthesis procedure that exploits model knowledge to develop tests that are provably realizable and maximally difficult.}
    \label{fig:title_fig}
    \vspace{-0.2 in}
\end{figure}

\subsection{Related Work}
\label{sec:related_work}
A growing area in the test and evaluation literature (T\&E) centers on the design and generation of tests utilizing a simulator of the system-under-test (SUT)~\cite{kapinski2016simulation}.  Here, a test corresponds to a specific environment setup in which test engineers evaluate the SUT's ability to satisfy a specified set of behaviors - the system objective.  To facilitate a rigorous objective satisfaction analysis, these objectives are oftentimes expressed as linear or signal temporal logic specifications which come equipped with specification satisfaction methods~\cite{pnueli1981linear_temporal,maler2004stl_foundational}.  For signal temporal logic specifically, each specification comes equipped with a robustness measure - functions over state signal trajectories whose positive evaluation corresponds to specification satisfaction.  To that end, some current work focuses on developing smoother robustness measures and using them as specification satisfaction monitors for real-time adaptation~\cite{bortolussi2014specifying,deshmukh2017robust,maler2013monitoring,hauer2019fitness,nivckovic2018amt}.

The test-synthesis question stems naturally from the existence of such robustness measures.  More specifically, provided a quantifiable set of phenomena that can frustrate specification satisfaction, the motivating question asks whether one can determine phenomena that minimize this robustness measure.  Here, decreasing robustness indicates increasing test difficulty as if a system has negative robustness with respect to specification satisfaction, then it has failed to satisfy this specification - its objective~\cite{maler2004stl_foundational}.  Each of the pre-eminent tools for automated test synthesis/falsification of system simulators expresses this parameter search as an optimization problem - S-Taliro \cite{annpureddy2011staliro}, Breach \cite{donze2010breach}, and more recently, VerifAI in conjunction with SCENIC \cite{dreossi2019verifai}.  Indeed there has also been a wealth of work using these tools to generate difficult tests of system behavior for multiple systems \cite{tuncali2016utilizing, fainekos2012verification, hoxha2014using, fremont2020formal, fremont2020scenic, tran2020nnv}.  There has also been work on optimally generating tests for specific controllers, independent of these tools \cite{abbas2013probabilistic, dreossi2015efficient, althoff2018automatic, klischat2019generating}.

The overarching goal of test generation, however, is to uncover problematic true system behavior without exhaustively testing the true system.  As such, there has been some work aimed at taking difficult simulator tests and realizing them on the real SUT \cite{fremont2020real_life}.  However, significantly more work aims at adapting the aforementioned tools to generate successively harder tests of real-system behavior \cite{ghosh2018verifying,gambi2019automatically,wheeler2019critical,gangopadhyay2019identification}.  As such, these tests fall into two categories.  Namely, they are \textit{controller-specific} and \textit{static} test cases insofar as they usually identify one parameter in a parametric set of disturbances that can frustrate system specification satisfaction for a given controller. 

Provided that the goal of test synthesis is to determine difficult tests of system behavior, we had posited earlier in~\cite{akella2020formal} that this study can and should be done controller-independent, as what is difficult for the system to achieve should be independent of the controller used to steer it.  Furthermore, while static tests can uncover problematic system behavior, a time-varying environment might identify more problematic phenomena.  For example, time-varying attacks on a system's control architecture might discover more troubling behavior than a static obstacle placed in front of a car.  As such, we endeavor to develop a formal method for generating adversarial, controller-agnostic, and time-varying tests of safety-critical system behavior.

\subsection{Summary of Contributions}
Our contribution is four-fold and will be itemized as follows:
\begin{itemize}
    \item First, for both continuous and discrete-time systems, we develop game-theoretic, adversarial test-synthesis procedures based on control barrier functions and timed reach-avoid specifications.  Subject to some assumptions on the system to be tested, these techniques satisfy the following criteria:
    \begin{itemize}
        \item they are guaranteed to produce a realizable test of system behavior, and
        \item the synthesized tests are provably the most difficult test of system behavior at that system state.
    \end{itemize}
    \item Second, we extend our continuous-time analysis to develop a test-synthesis procedure for tests that perturb the system dynamics directly.  We similarly prove the existence and maximal difficulty of tests in this setting as well.
    \item Third, we extend our discrete-time analysis to develop a predictive test-synthesis procedure that is guaranteed to be realizable and maximally difficult.
    \item Fourth, we extend both our continuous and discrete-time test-synthesis techniques to the scenario when the feasible space of tests may be time-varying or otherwise constrained.  In this setting, we also prove that our method is guaranteed to produce realizable and maximally difficult tests of system behavior.
    \item Finally, we showcase the results of our test-generation procedure for each case.  For the unconstrained examples, we showcase our results in simulation, mention some deficiencies resolved in our constrained extensions, and showcase the results of our constrained test-synthesis procedure by testing a quadruped robot's ability to navigate while avoiding moving obstacles.
\end{itemize}
For context, we restrict our study to timed reach-avoid specifications, a subset of signal temporal logic specifications, as they are commonly used in the controls literature to represent basic robotic objectives, \textit{e.g.} reach a goal and avoid obstacles along the way~\cite{lindemann2017robust,lindemann2018control,lindemann2019decentralized}.

\subsection{Organization}
To start, Section~\ref{sec:preliminaries} will detail some necessary background information - control barrier functions and timed reach-avoid specifications - and Section~\ref{sec:prob_form_state} will set up and formally state the problem under study in this paper.  Then, Section~\ref{sec:continuous} details our adversarial test-synthesis procedure in the continuous setting, and Section~\ref{sec:cont_examples} illustrates our main results in continuous time through a simple example.  Likewise, Section~\ref{sec:discrete} details our results in the discrete setting , and Section~\ref{sec:disc_examples} illustrates these results through a simple example as well.  Finally, Section~\ref{sec:extensions} extends the results of both of the prior sections by developing a test-synthesis procedure that has similar guarantees on existence and difficulty in a constrained test-synthesis scenario where the space of feasible tests varies with time or the system state.  Likewise,  Section~\ref{sec:constrained_hardware} details the application of our test-synthesis procedure to providing difficult tests for a quadrupedal system in its satisfaction of a simple objective.  Before moving to the next section, we will briefly define some Notation.

\sectionspacing
\newidea{Notation:} The set $\mathbb{Z}_{+} = \{0,1,2,\dots\}$.  A function $\alpha: (-b,a) \to \mathbb{R} \cup \{-\infty,\infty\}$ where $a,b \in \mathbb{R}_{++}$, is an extended class-$\kappa$ function $\kappa_e$ if and only if $\alpha(0) = 0$, and for $r > s$ $\alpha(r) > \alpha(s)$. For any set $A$ define $|A|$ to be the cardinality of $A$, \textit{i.e.} the number of elements in $A$.  For a function $f: X \to Y$ we say $f \in C^k(X)$ if $f$ has at least $k$ (partial) derivatives and its $k$-th (partial) derivative(s) is/are continuous, \textit{i.e.} for $x \in \mathbb{R}^2$, $f = x^Tx \in C^2(\mathbb{R}^2)$.  For a set $A$, $2^A$ is the set of all subsets of $A$, \textit{i.e.} $2^A = \{B~|~B \subseteq A\}$.

\section{Preliminaries}
\label{sec:preliminaries}
This section will provide a brief introduction to some necessary topics.  We will first introduce control barrier functions - a control technique used to ensure safety in control-affine systems.  Then, we will introduce timed reach-avoid specifications - a mathematical formalism to describe system objectives.

\subsection{Control Barrier Functions}
\label{sec:cbf}
Inspired by their counterparts in optimization (see Chapter 3 of \cite{forsgren2002interior}), control barrier functions are a modern control tool to ensure safety in safety-critical systems that are control-affine.  For a study of control barrier functions in the continuous setting, please see~\cite{ames2016control}, and in the discrete setting, please see~\cite{agrawal2017discrete}.  The main purpose of control barrier functions $h: \mathbb{R}^n \to \mathbb{R}$ are to ensure forward invariance of their $0$-superlevel sets $\mathcal{C}$:
\begin{align}
    \mathcal{C} = \{x \in \mathbb{R}^n~|~h(x) \geq 0\},~
    \partial \mathcal{C} = \{x\in \mathbb{R}^n~|~h(x) = 0\}.
\end{align}
As will be shown later, control barrier functions and their $0$-Superlevel sets offer one way of quantifying specification satisfaction for robotic systems.  This quantification will be instrumental in our adversarial approach for both the continuous and discrete settings.

\sectionspacing
\newidea{Continuous Setting:} A study in continuous-time control barrier functions typically assumes a control-affine nonlinear control system as its control system abstraction~\cite{ames2016control}:
\begin{equation}
    \label{eq:nom_sys}
    \dot x = f(x) + g(x) u, \quad x \in \mathcal{X} \subseteq, \mathbb{R}^n,~u\in\mathcal{U}\subseteq \mathbb{R}^m.
\end{equation}
Then, a control barrier function $h$ is defined as any function satisfying a specific inequality over its Lie Derivatives.  Specifically, for a function $h:\mathbb{R}^n \to \mathbb{R}$, its Lie Derivatives $L_fh, L_gh$ with respect to the system dynamics~\eqref{eq:nom_sys} are:
\begin{equation}
    L_fh(x) = \frac{\partial h}{\partial x}(x) f(x),\quad L_gh(x) = \frac{\partial h}{\partial x}(x)g(x).
\end{equation}
Then, the definition of control barrier functions is as follows.
\begin{definition} (Adapted from Definition 5 in~\cite{ames2016control})
\label{def:continuous_cbf}
For the nominal control-affine system~\eqref{eq:nom_sys}, a \textit{control barrier function} $h:\mathbb{R}^n \to \mathbb{R}$ satisfies the following condition, $\forall~x\in\mathcal{X}$ and some $\alpha \in \kappa$:
\begin{equation}
    \max_{u \in \mathcal{U}}\left[L_fh(x) + L_gh(x) + \alpha(h(x)) \geq 0 \right],
\end{equation}
\end{definition}
Here, we note that as initially defined in \cite{ames2016control}, the CBFs defined above are termed \textit{zeroing control barrier functions}. As we do not deal with \textit{reciprocal control barrier functions}, we will simply refer to these as control barrier functions in what will follow.  Additionally, we will follow convention and denote continuous control barrier functions as simply control barrier functions as their discrete counterparts are specifically termed discrete control barrier functions.

\sectionspacing
\newidea{Discrete Setting:} In the discrete setting, as per~\cite{agrawal2017discrete}, the nominal control system abstraction is a discrete, nonlinear control system, whose update equation is valid $\forall~k\in\mathbb{Z}_+$:
\begin{equation}
    \label{eq:discrete_system}
    x_{k+1} = f(x_k,u_k),~x_k \in \mathcal{X} \subset \mathbb{R}^n,~u_k\in\mathcal{U}\subset \mathbb{R}^m.
\end{equation}
The definition of a discrete control barrier function stems immediately from the assumed system dynamics~\eqref{eq:discrete_system}.
\begin{definition}
\label{def:discrete_cbf}
(Adapted from Definition 2 in~\cite{ahmadi2019safe}) For the discrete-time control system~\eqref{eq:discrete_system}, a \textit{discrete control barrier function} $h:\mathbb{R}^n \to \mathbb{R}$ satisfies the following condition for some class-$\kappa$ function $\alpha$ such that $\alpha(r) < r$:
\begin{equation}
    \label{eq:discrete_cbf_inequality}
    \exists~u_k \in \mathcal{U} \suchthat h(x_{k+1}) - h(x_k) \geq -\alpha(h(x_k))~\forall~x\in \mathcal{X}.
\end{equation}
\end{definition}

As mentioned earlier, the main purpose of control barrier functions in either setting is to ensure forward invariance of their $0$-superlevel sets $\mathcal{C}$.  Typically, the application of an optimization-based controller, constrained by the respective inequalities in either definition, suffices to ensure this forward invariance~\cite{ames2016control,agrawal2017discrete}.  In what will follow, we will make a one-to-one correspondence between these $0$-superlevel sets and the truth regions for predicates in our timed reach-avoid specifications, allowing for quantification of specification satisfaction.

\subsection{Timed Reach-Avoid Specifications}
\label{sec:STL}
Timed reach-avoid specifications express common robotic objectives, through the use of logical predicates and combinations thereupon.  As they are a specific subset of signal temporal logic specifications, we will provide a brief overview of the latter and then restrict it to the former.  As signal temporal logic deals with signals, we define a signal $s$ as a function that maps from time to $\mathbb{R}^n$, \textit{i.e.} $s: \mathbb{R}_{+} \to \mathbb{R}^n$.  Then, we define logical predicates $\mu$ with respect to a predicate function $b: \mathbb{R}^n \to \mathbb{R}$:
\begin{equation}
    \llbracket \mu \rrbracket \triangleq \{x\in\mathbb{R}^n~|~b(x) \geq 0\}, \quad \mu(x) = \true~\mathrm{iff}~x \in \llbracket \mu \rrbracket. \label{eq:STL_predicate}
\end{equation}
The set of all predicates $\mathcal{A}$ is closed under logical combinations, $\wedge$ (and), $\lor$ (or), $\neg$ (negation), \textit{i.e.}
\begin{align}
    \mu \in \mathcal{A} & \iff \neg \mu \in \mathcal{A}, \\
    \mu_1,\mu_2 \in \mathcal{A} & \iff \mu_1 \wedge \mu_2 \in \mathcal{A}~\mathrm{and}~\mu_1 \lor \mu_2 \in \mathcal{A}.
\end{align}

From these predicates, system behavior is specified through \textit{specifications} $\psi$ (here the $"|"$ demarcates different instances in how a specification $\psi$ may be composed):
\begin{equation}
    \psi \triangleq \true~|~\mu~|~\neg\psi~|~\psi_1\wedge\psi_2~|~\psi_1\lor\psi_2~|~\psi_1\until_{[a,b]}\psi_2.
\end{equation}
Here, $\psi_1, \psi_2$ are specifications themselves, and the bounded time until operator $\until_{[a,b]}$ reads as $\psi_1 \until_{[a,b]} \psi_2$: $\psi_1$ is to be true until $\psi_2$ becomes true, and $\psi_2$ should become true at some time $t' \in [t+a,t+b]$, for some $a,b \in \mathbb{R}_{+},b>a,$ and initial time $t$.  We denote that a signal $s$ satisfies a specification $\psi$ at time $t$ as $(s,t) \models \psi$.  Here, $\models$ is the satisfaction relation and is inductively defined as follows:
\begin{align}
    & (s,t) \models \mu \iff \mu(s(t)) = \true, \\
    & (s,t) \models \neg \psi \iff (s,t) \not \models \psi, \\
    & (s,t) \models \psi_1 \lor \psi_2 \iff (s,t) \models \psi_1 \lor (s,t) \models \psi_2, \\
    & (s,t) \models \psi_1 \wedge \psi_2 \iff (s,t) \models \psi_1 \wedge (s,t) \models \psi_2, \\
    & (s,t) \models \psi_1 \until_{[a,b]} \psi_2 \iff \exists~t^*\in[t+a,t+b]\suchthat \\
    & \quad \left((s,t') \models \psi_1~\forall~t' \in[t,t^*]\right) \wedge \left((s,t^*) \models \psi_2 \right).
\end{align}

As mentioned previously, we will restrict our analysis to timed reach-avoid specifications which make use of two commonly used operators, $\F$uture and $\G$lobal, which are defined as follows:
\begin{equation}
    \F_{[a,b]}\mu = \true \until_{[a,b]} \mu, \quad \G_{[a,b]}\mu = \neg (\F_{[a,b]}\neg \mu).
\end{equation}
Here, $\F_{[a,b]}\mu$ indicates a specification where $\mu$ is to be true at some time in the future $t' \in [t+a,t+b]$ with respect to some initial time $t$.  Likewise, $\G_{[a,b]}\mu$ indicates a specifications where $\mu$ is to be true for all times $t' \in [t+a,t+b]$ given an initial time $t$. Oftentimes, safety specifications require continued satisfaction of a predicate $\mu$.  In these cases, we will use the shorthand $\G_{\infty} \mu$.

Coupling control barrier functions and timed reach-avoid specifications then requires choosing the predicate function $b$ for the predicate $\mu$ to be a control barrier function $h$.  This results in the predicate's truth region $\llbracket \mu \rrbracket = \mathcal{C}$, the $0$-super-level set of the same control barrier function $h$.  Then, specifications like $\F_{[0,b]} \mu$ or $\G_{\infty} \mu$ can be quantitatively analyzed through finite time positivity of the associated control barrier function (for $\F_{[0,b]} \mu$) or continued positivity of the associated control barrier function (for $\G_{\infty} \mu$).  For context, this is one way of coupling control barrier functions and subsets of signal temporal logic, though this intersection has also been studied more broadly, \textit{e.g.} see~~\cite{lindemann2017robust, lindemann2018control, lindemann2019robust}.  Finally, it is well known that every signal temporal logic specification, and therefore every timed reach-avoid specification, has a robustness measure $\rho$ that provides a quantitative analysis similar to that provided by control barrier functions in our analysis~\cite{madsen2018metrics,maler2004stl_foundational}.  However, these tend to be non-differentiable, and the differentiability of these control barrier functions will be a key part of our analysis.

\section{Problem Formulation and Statement}
\label{sec:prob_form_state}
In this section, we will first state some common definitions and assumptions that will be used throughout the paper.  Then, we will use these definitions to formally state the problem under study.

\subsection{Definitions and Common Assumptions}
\label{sec:definitions}
To generate a test-synthesis procedure, we need to first formally define a test and the system specifications we aim to test through our procedure.  We will start with the latter and use it to formalize the former.  As mentioned, we restrict our analysis to timed reach-avoid specifications as these are commonly used specifications in the robotics literature - their general form will follow~\cite{sadigh2016safe, lindemann2018control, raman2014model, haghighi2019control}:
\begin{equation}
    \label{eq:spec}
    \psi = \F_{[0,t_{\max}]} \mu \wedge_{j \in \mathcal{J}} \G_\infty \omega_j.
\end{equation}
Here, $\mathcal{J} = \{1,2,\dots,|\mathcal{J}|\}$ is merely a set of indices demarcating different safety objectives the system is to maintain, and $t_{\max}$ denotes the maximum time by which the desired objective $\mu$ is to be achieved.  This leads to the first assumption in our work, which is as follows.
\begin{assumption}
\label{assump:satisfiability}
We assume the specification $\psi$~\eqref{eq:spec} is satisfiable.
\end{assumption}
\noindent  Assumption~\ref{assump:satisfiability} prevents against mutually exclusive/conflicting safety objectives in the overarching specification form~\eqref{eq:spec}.  We make this assumption as from a test-generation perspective, determining tests for a specification that could never be satisfied is irrelevant.  Any test is hard as the system could never satisfy its objective by construction.  That being said, determining the satisfiability of a given signal temporal logic specification is the subject of current work~\cite{bae2019bounded}. To elucidate the specifications of this form, we will provide an example.
\begin{example}
\label{ex:turtle}
Consider an idealized unicycle system on a 2-d plane.  Let this system's goal be to navigate to a region defined as the $0$-superlevel set of a function $h^F:\mathbb{R}^2 \to \mathbb{R}$ while avoiding a set of obstacles defined as the conjunction of the $0$-sublevel sets of multiple other functions $h^G_j:\mathbb{R}^2 \to \mathbb{R},~\forall~j \in \mathcal{J} = \{1,2,\dots\}$.  See Figure~\ref{fig:example_1} for an illustration.  Then, the system's specification $\psi = \F_{\infty} \mu \wedge_{j \in \mathcal{J}} \G_{\infty} \neg \omega_j$ where $\llbracket \mu \rrbracket = \mathcal{C}_{h^F}$ and $\llbracket \neg \omega_j \rrbracket = \mathcal{C}_{h^G_j}$.  For $\psi$ to satisfy Assumption~\ref{assump:cbf_stl}, there must be at least one set of obstacle and goal locations wherein the robot is capable of navigating to the goal while avoiding all obstacles, \textit{i.e.} the goal shouldn't always be at infinity or always encapsulated by obstacles.
\end{example}

\begin{figure}
    \centering
    \includegraphics[width = 0.49\textwidth]{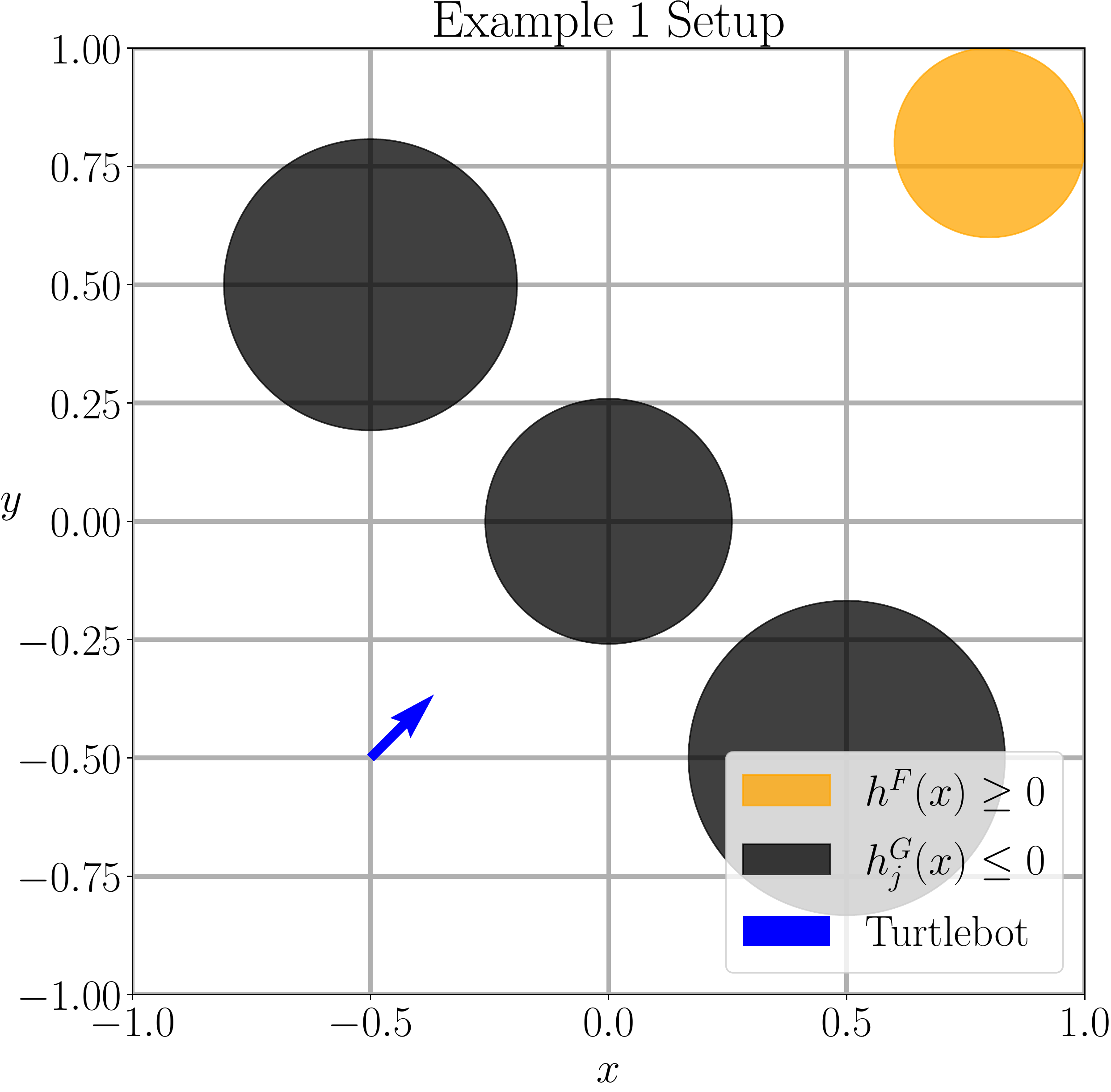}
    \caption{Example setup for Example~\ref{ex:turtle}.  The agent is shown via the blue arrow, the obstacles via the black circles of varying sizes, and the goal via the golden circle.  The corresponding 0-superlevel and sublevel sets follow the same color scheme as shown in the legend.}
    \label{fig:example_1}
    \vspace{-0.2 in}
\end{figure}

Continuing with Example~\ref{ex:turtle}, the most natural test of this agent's behavior would be to see if it could satisfy its specification irrespective of the locations of obstacles in its environment.  The specific setup of obstacle locations would correspond to a specific test.  This notion underlies how we will formalize tests.  To start, we will formally define the system's environment.
\begin{definition}
\label{def:environment}
The \textit{environment} $E$ is the state of the world in which the system operates including the state of the system itself, \textit{e.g.} the cave in which a robot is traversing coupled with any motor failures the robot may have suffered, the airspace in which a jet flies along with any engine failures, \textit{etc}.  The state of the environment will be represented through the \textit{environment state vector} $x_E$.
\end{definition}
Here, we note that the environment's state vector $x_E$ may be incomprehensibly large, indeed even infinite, for real-life systems.  As such, if we wanted to test and evaluate a system's ability to satisfy its specification independent of the environment it might face, we need some way of modeling at least part of the environment to understand its effects on system specification satisfaction.  To note, this does not imply that we know and can model everything about our system's intended environment.  Rather, we assume that there exists some subset of this environment that we can model and can effectively use to test system behavior, \textit{e.g.} we can model wind speeds a drone might face during operation and will test our drone by requiring it to satisfy its objective when subject to a set of wind-speeds that we can realize through our model.  As such, we will make the following definition segmenting the system's environment into the set of modelable phenomena and un-modelable phenomena.
\begin{definition} 
\label{def:valid_tests}
The state $x_E$ of the environment $E$ can be segmented into the system state $x$, a set of modelable phenomena $d \in \mathcal{D}$, and a set of un-modelable phenomena $w \in \mathbb{W}$, \textit{i.e.} $x_E^T = [x^T, d^T, w^T]^T$.  The space of modelable phenomena $\mathcal{D}$ is the \textit{feasible test space} and each $d \in \mathcal{D}$ is a \textit{test parameter vector}.
\end{definition}

For Example~\ref{ex:turtle} then, the environment $E$ is the 2-d plane within which the agent operates and the obstacles lie.  Here, the environment state $x_E$ can be segmented into the agent's state, modelable obstacle locations $d \in \mathcal{D}$, and un-modelable table friction, signal delays, motor constants, \textit{etc}, all of which comprise $w \in \mathbb{W}$.  As motivated earlier, an example test would correspond to placing the obstacles at known locations \textit{apriori} and then allowing the agent to try navigating to its goal.  This motivates our definition of a test that will follow.

\begin{definition}
\label{def:test}
A \textit{test} is a specific environment setup modeled by a specific \textit{test parameter vector} $d \in \mathcal{D}$.
\end{definition}

Per this definition, the outcome of the agent navigating to its goal need not be the same for two similar tests - this is primarily due to the lack of knowledge of $w \in \mathbb{W}$.  However, as these variables are un-modelable we cannot restrict them during a test.  As a result, when we formalize our test-generation procedure, we will only focus on determining a suitable choice of test parameter vector $d$, as the unknown phenomena $w$ are by definition, unknown.  Also, we would expect that running the same test twice might yield slightly different outcomes each time.  This scenario oftentimes happens in reality, \textit{e.g.} running the same robot twice and experiencing slightly different behaviors each time, due to friction, battery power loss, \textit{etc}.  We will extend Example~\ref{ex:turtle} to better illustrate Definition~\ref{def:test}.
 
\begin{example}
\label{ex:turtle_2}
In the setup in Example~\ref{ex:turtle}, let each obstacle be defined as the interior of a circle and define $\llbracket \omega_j \rrbracket = \{x ~|~ \|x - d_j\| \leq r_j\}$ for each $j \in \mathcal{J}$.  Then, define a vector $d = [d_1, r_1, d_2, r_2, \dots ] \in \mathbb{R}^{3|\mathcal{J}|}$.  Each specific vector $d$ constitutes a set of known obstacle locations.  Requiring the agent to navigate to its goal in an environment $E$ formed from that choice of obstacle locations $d$ is an example test.
\end{example}

Our goal is to constructively determine adversarial, time-varying tests of system behavior with respect to satisfaction of a timed reach-avoid specification in a controller-agnostic fashion.  We first note that as a direct consequence of Definition~\ref{def:test}, we can define a time-varying test as one where the test parameter vector $d$ varies with time, \textit{i.e.} $\testsynth: \mathbb{R}_{+} \to \mathcal{D}$.  For Example~\ref{ex:turtle_2}, a time-varying test would correspond to a scenario where the obstacles are moving while the agent is navigating to its goal.  To facilitate the development of such a procedure, we require one assumption on the existence of control barrier functions corresponding to the timed reach-avoid specification $\psi$ that the system is to satisfy.  This assumption permits us to make a relation between the satisfaction of the system specification $\psi$~\eqref{eq:spec} and the positivity of their corresponding control barrier functions.

\begin{assumption}
\label{assump:cbf_stl}
Let $\mu$, $\omega_j$ be predicates comprising the timed reach-avoid specification $\psi$~\eqref{eq:spec}.  We assume there exists $\forall~j\in\mathcal{J}$,
\begin{subequations}
\begin{align}
    h^G_j: \mathbb{R}^n \times \mathcal{D} \to \mathbb{R}&\suchthat h^G_j(x,d) \geq 0 \iff x \in \llbracket \omega_j \rrbracket, \quad \label{eq:global_barrier}\\
    h^F: \mathbb{R}^n \times \mathcal{D} \to \mathbb{R}& \suchthat h^F(x,d) \geq 0 \iff x \in \llbracket \mu \rrbracket \label{eq:future_barrier}
\end{align}
\end{subequations}
where each $h^G_j$ and $h^F$ are (discrete) control barrier functions.
\end{assumption}

For context, Assumption~\ref{assump:cbf_stl} is not too restrictive.  It builds off prior work that constructs control barrier functions for single/multi-agent systems subject to signal temporal logic specifications more broadly~\cite{lindemann2017robust,lindemann2018control,lindemann2019decentralized,lindemann2019robust,lindemann2020efficient}.  As timed reach-avoid specifications are a subset of signal temporal logic specifications, and the existence of control barrier functions for these types of signal temporal logic tasks has been determined \textit{apriori}, we will simply assume their existence for the time being.
With these definitions and assumptions, we can formally state the problem under study in this paper.

\subsection{Problem Statement}
\label{sec:problems}
As mentioned, our goal is to develop an adversarial, time-varying test-generation procedure for safety-critical systems in a controller-agnostic fashion.  Ideally, we would also like to develop a procedure that works for both continuous and discrete-time systems.  Furthermore, we would like for our generated tests to be maximally difficult with respect to a difficulty metric.  As a first step in developing a concrete problem statement, we would like to formalize a test difficulty metric - some function that is minimized at a given state $x$ by the hardest test vector $d$ at that state $x$.  To that end, we will first define a space of feasible inputs $\mathcal{U}(x,d)$.
\begin{definition}
\label{def:feasible_inputs}
The \textit{instantaneous feasible input space} $\mathcal{U}(x,d) \subseteq \mathcal{U}$ is the space of inputs satisfying the following conditions: if $\forall~t\geq 0$ (respectively all $k \in \mathbb{Z}_+$), the control inputs $u \in \mathcal{U}(x(t),d)$ (respectively $u \in \mathcal{U}(x_k,d)$), the resulting state trajectory $x(t)$ (respectively $x_k$) satisfies $\psi = \wedge_{j \in \mathcal{J}}\G_{\infty}\omega_j$, \textit{i.e.} $(x,0) \models \wedge_{j \in \mathcal{J}}\G_{\infty}\omega_j$.
\end{definition}
\noindent Effectively, the feasible input space $\mathcal{U}(x,d)$ is the space of inputs that steer the system in satisfaction of its safety specifications $\G_{\infty}\omega_j$ in equation~\eqref{eq:spec}.  On a related note, we will also require a function that discriminates between feasible input choices taken by the system in satisfaction of its desired objective $\F_{[0,t_{\max}]}\mu$.  We will define this action discriminator function $v$ as follows.
\begin{definition}
\label{def:heading}
The \textit{action discriminator} is a function $v: \mathcal{X} \times \mathcal{D} \times \mathcal{U} \to \mathbb{R}$ with $\mathcal{D}$ the feasible test space as per Definition~\ref{def:valid_tests} that satisfies the following condition: if $\forall~t \geq 0$ (respectively all $k \in \mathbb{Z}_+$), the control inputs $u$ are such that $v(x(t),d,u) \geq 0$ (respectively $v(x_k,d,u) \geq 0$), the resulting state trajectory $x(t)$ (respectively $x_k$) satisfies $\psi = \F_{[0,t_{\max}]}\mu$, \textit{i.e.} $(x,0) \models \F_{[0,t_{\max}]}\mu$.
\end{definition}
\noindent Now, we can use the feasible input space $\mathcal{U}(x,d)$ and our action discriminator $v$ to formalize our controller-agnostic difficulty measure.
\begin{definition}
\label{def:difficulty_metric}
An \textit{instantaneous difficulty metric} $M:\mathcal{X} \times \mathcal{D} \to \mathbb{R}$ has a well-defined minimizer $\forall~x \in \mathcal{X}$,
\begin{equation}
    \testsynth(x) = \argmin_{d \in \mathcal{D}}~M(x,d),
\end{equation}
that satisfies one of the two conditions below for the same state $x$,
\begin{gather}
    \mathcal{U}(x,\testsynth(x)) = \varnothing,~\mathrm{or}~\testsynth(x) = \argmin\limits_{d \in \mathcal{D}}~\max\limits_{u \in \mathcal{U}(x,d)}~v(x,d,u).
\end{gather}
Here, $\mathcal{U}(x,d)$ is the system's feasible input space as per Definition~\ref{def:feasible_inputs}, and $v$ is the action discriminator as per Definition~\ref{def:heading}.
\end{definition}
\noindent Intuitively then, our difficulty metric $M$ quantifies test difficulty through minimizing the "maximum possible increment" the system could take towards satisfying its desired objective $\mu$ while maintaining its safety specifications $\omega_j$.  Here, the "maximum possible increment" to satisfaction is provided through the action discriminator $v$ which closely resembles the robustness measures traditional to signal temporal logic~\cite{madsen2018metrics,maler2004stl_foundational,maler2013monitoring}.  Effectively, our action discriminator is an instantaneous version of this robustness measure in that it assigns positive values to those system actions $u$ which bring the system closer to satisfying its objective.  It is this instantaneous notion that allows us to develop time-varying tests.  As such, our formal problem statement will follow.
\begin{problem}
For both control system abstractions~\eqref{eq:nom_sys} or~\eqref{eq:discrete_system}, let the system's operational specification $\psi$ satisfy equation~\eqref{eq:spec}.  Then, for either system abstraction, develop (perhaps different) \textit{adversarial test-synthesis procedure(s)} $\testsynth: \mathcal{X} \to \mathcal{D}$ that
\begin{itemize}
    \item are guaranteed to produce a realizable test, \textit{i.e.} $\forall~x\in\mathcal{X},~\exists~d\in\mathcal{D}$ such that $d = \testsynth(x)$;
    \item are the most difficult tests of system behavior at that state, independent of control input, \textit{i.e.} $\testsynth(x) = \argmin_{d \in \mathcal{D}} M(x,d)$ for some difficulty measure $M$ that satisfies definition~\ref{def:difficulty_metric}.
\end{itemize}
\end{problem}
\noindent Here, we note that while we have defined our specifications to be evaluated over continuous time, they can still be used to express specifications for discrete-time systems~\cite{raman2014model,raman2015reactive,lindemann2017robust,lindemann2019robust}.  In the discrete case, time is indexed by $k \in \mathbb{Z}_+$ as opposed to the continuous analog where time is measured on the positive reals.

\section{Continuous-Time Test Generation}
\label{sec:continuous}
In this section, we will state and prove our main result, the development of an adversarial, time-varying test-synthesis procedure for continuous, control-affine control systems of the form~\eqref{eq:nom_sys} subject to timed reach-avoid specifications $\psi$ of the form in~\eqref{eq:spec}.  We will briefly describe the overarching methodology behind our approach, state the developed minimax problem for test synthesis, and end with two theorems regarding its use.

\sectionspacing
\newidea{Overarching Idea:} Based on the recursive definition of the satisfaction relation $\models$ in Section~\ref{sec:STL}, a controller for the nominal system~\eqref{eq:nom_sys} guarantees system satisfaction of the timed reach-avoid specification $\psi$ in~\eqref{eq:spec} if and only if, $\forall~ t \geq 0,$
\begin{equation}
    \label{eq:continuous_specification_satisfaction}
    x(t) \in \cap_{j \in \mathcal{J}}~\llbracket \omega_j \rrbracket,~\mathrm{and~}
    \exists~t' \in [0,t_{\max}] \suchthat x(t') \in \llbracket \mu \rrbracket.    
\end{equation}
\noindent Based on Assumption~\ref{assump:cbf_stl}, the requirement in equation~\eqref{eq:continuous_specification_satisfaction} translates to the following statement $\forall~t \geq 0$ and with $\mathcal{C}$ as the $0$-superlevel set for the control barrier function shown as a subscript:
\begin{equation}
    \label{eq:continuous_cbf_end}
    x(t) \in \cap_{j \in \mathcal{J}}~\mathcal{C}_{h^G_j},~\mathrm{and~}
    \exists~t' \in [0,t_{\max}] \suchthat x(t') \in \mathcal{C}_{h^F}.
\end{equation}
\noindent Without loss of generality, we can assume that at $t = 0$,
\begin{equation}
    \label{eq:continuous_cbf_start}
    x(0) \in \cap_{j \in \mathcal{J}}~\mathcal{C}_{h^G_j},~\mathrm{and~}
    x(0) \not \in \mathcal{C}_{h^F},
\end{equation}
as otherwise, the system would either never be able to satisfy $\psi$ - as it started in an unsafe region, \textit{i.e.} $x(0) \not \in \llbracket \omega_j \rrbracket$ for at least one $j \in \mathcal{J}$- or it would satisfy $\psi$ by remaining stationary.  Neither case constitutes interesting test cases.  Hence, as each $h^G_j$ and $h^F$ is a control barrier function, one way of transitioning from the starting condition~\eqref{eq:continuous_cbf_start} to the end condition~\eqref{eq:continuous_cbf_end} is for the controller to satisfy the following inequalities for some functions $\alpha,\alpha_j \in \kappa_e$ and $\tau > 0$,
\begin{align}
    \dot h^F(x(t),d,u) & \geq -\alpha\left(h^F(x(t),d)\right) + \tau,~\forall~t \in [0,T], \\
    \dot h^G_j(x(t),d,u) & \geq -\alpha_j\left(h^G_j(x(t),d)\right),~\forall~ t \geq 0.
\end{align}
The above conditions provide us a way of generating quantifiably adversarial tests - generate tests that are designed to minimize satisfaction of these inequalities.

\subsection{Statement of Continuous-Time Results}
\label{sec:cont_main_results}
To formalize this satisfaction minimization idea mentioned prior, we will first specify a set of feasible inputs $\mathcal{U}(x,d)$ and an action discriminator function $v$ as follows (with some $\tau > 0$):
\begin{align}
    \hspace{-0.1 in} \mathcal{U}(x,d) & = \left\{u\in\mathcal{U}~\bigg|~\dot h^G_j(x,d,u) \geq -\alpha_j\left(h^G_j(x,d)\right),~\forall~j\right\}, \label{eq:feasible_input} \\
    \hspace{-0.1 in} v(x,d,u) & = \dot{h}^F(x,d,u) - \tau. \label{eq:cont_act_discrim}
\end{align}
Here, we first note that $\mathcal{U}(x,d)$ above identifies those inputs $u \in \mathcal{U}$ that satisfy the CBF condition expressed in Definition~\ref{def:continuous_cbf} for the barrier functions $h^G_j$.  As such, by the work done in~\cite{ames2016control,li2018formally} and Assumption~\ref{assump:cbf_stl}, we know that $\mathcal{U}(x,d)$ is a valid feasible input space as per Definition~\ref{def:feasible_inputs}.  Likewise, for some $\tau > 0$, $v$ is also a valid action discriminator as per Definition~\ref{def:heading}.  To briefly show that this holds for $\mathcal{U}(x,d)$, if the starting condition in equation~\eqref{eq:continuous_cbf_start} and Assumption~\ref{assump:cbf_stl} hold, then by the work done in~\cite{ames2016control}, if all inputs $u$ are chosen such that $u \in \mathcal{U}(x(t),d)~\forall~t\geq 0$, then,
\begin{equation}
    x(t) \in \cap_{j \in \mathcal{J}}~\mathcal{C}_{h^G_j}~\forall~t \geq 0 \implies (x,0) \models \wedge_{j \in \mathcal{J}}\G_{\infty}\omega_j.
\end{equation}
\noindent As a result, $\mathcal{U}(x,d)$ is indeed an instantaneous space of feasible control inputs as per Definition~\ref{def:feasible_inputs}, and via a similar chain of logic, we can show $v$ is also a feasible action discriminator.  As we want to keep our test-synthesis framework controller-agnostic, we will initially propose the following minimax problem over all feasible inputs in equation~\eqref{eq:feasible_input} as our test synthesizer:
\begin{equation}
    \label{eq:proposed_synthesis_continuous}
    \testsynth(x) = \argmin_{d \in \mathcal{D}}\max_{u \in \mathcal{U}(x,d)}~\dot h^F(x,d,u).
\end{equation}

\newidea{Remark on Local/Global Conservatism:}  However, this proposed synthesis technique has two shortcomings.  First, our proposed synthesis technique may suffer from a local-global problem - that by determining a worst-case test at a given state $x$ we find a time-varying test sequence $\testsynth(x(t))$ that is only locally optimal, \textit{i.e.} locally difficult but not globally difficult.  Rectifying this shortcoming while maintaining our controller-agnostic approach would require us to optimize over control sequences in the inner maximization problem in~\eqref{eq:proposed_synthesis_continuous}.  As we assume nonlinear dynamics in~\eqref{eq:nom_sys}, however, this would result in a non-convex inner maximization problem wherein it would be difficult to determine the feasibility of any resulting minimax problem.  In the proposed case, however, we can guarantee feasibility as we will show later, and the tests generated are still interesting.  That being said, determination of tests that are globally difficult and verifying wholistic system behavior is the subject of current work~\cite{akella2022barrier,akella2022scenario, akella2022sample}.

The second shortcoming is that it may be the case that there exist test parameter vectors $d \in \mathcal{D}$ such that the feasible input space $\mathcal{U}(x,d) = \varnothing$, as we have made no effort to restrict against this scenario.  In these cases, the inner maximization problem would be ill-posed, frustrating any further analysis.  However, were there such a test parameter vector $d$, we would like to identify it as a worst-case test.  Indeed this is one of the conditions we used to define our difficulty metric in Definition~\ref{def:difficulty_metric}.  To facilitate analysis in the scenario where the feasible input space might be empty then, we will define a function $\mathcal{F}$ which filters a solution based on the emptiness (or lack thereof) of a provided set.  More accurately, for two scalars $\epsilon,\zeta \in \mathbb{R}$, an arbitrary set $A \subset \mathbb{R}^m$, and a vector $a \in \mathbb{R}^m$, define $\mathcal{F}$ as follows:
\begin{equation}
    \label{eq:filtering}
    \mathcal{F}(\epsilon,a,A,\zeta) = \begin{cases}
        \epsilon & \mbox{if}~a \in A, \\
        \zeta & \mbox{else}.
    \end{cases}
\end{equation}
Then, we will make one assumption on the system dynamics~\eqref{eq:nom_sys}, the feasible test space $\mathcal{D}$, and our control barrier functions.
\begin{assumption}
\label{assump:continous_assumption}
Both the state space $\mathcal{X}$ and the feasible test space $\mathcal{D}$ are compact, the input space $\mathcal{U}$ is a compact polytope in $\mathbb{R}^m$, and $h^F,h^G_j \in C^1(\mathcal{X} \times \mathcal{D})$.
\end{assumption}

For context, Assumption~\ref{assump:continous_assumption} is not that restrictive and we will explain why.  First, we restrict the space of feasible tests $\mathcal{D}$ to a compact set as we do not expect our test parameter vector $d$ to tend to $\pm \infty$.  Furthermore, if $d$ is bounded, we expect our realized test to be capable of taking values on the boundary.  For an example, consider Example~\ref{ex:turtle_2} where the obstacles are allowed to take center locations on the boundary of the hyper-rectangle in which they are confined.  Then, the assumptions of compactness on the state space $\mathcal{X}$ and polytopic compactness of the input space $\mathcal{U}$ are satisfied by most torque-bounded robotic systems. Finally, the assumption of continuity of the control barrier functions and their first partial derivatives is an easily satisfied restriction on the smoothness of our control barrier functions.  For an example, consider Example~\ref{ex:turtle} where this holds.

Then, we will define a minimum satisfaction value $m$ that meets the following inequality $\forall~d \in \mathcal{D}$ and $x \in \mathcal{X}$:
\begin{equation}
    \label{eq:min_satisfaction}
    m \leq \min_{u \in \mathcal{U},~x \in \mathcal{X},~d \in \mathcal{D}}~\dot{h}^F(x,d,u).
\end{equation}
While it is unclear at the moment whether such an $m$ exists, we will formally prove its existence in the proof for Theorem~\ref{thm:existence} to follow.  Now, we can formally state our test-synthesis procedure, with $\mathcal{F}$ as in~\eqref{eq:filtering}, $m$ as in equation~\eqref{eq:min_satisfaction}, and $\mathcal{U}(x,d)$ as per equation~\eqref{eq:feasible_input}.
\begin{equation}
    \label{eq:feedback_law}
    \testsynth(x) = \argmin_{d \in \mathcal{D}}~\max_{u \in \mathcal{U}}~\mathcal{F}\left(\dot h^F(x,d,u), u, \mathcal{U}(x,d), m\right).
\end{equation}
This leads to our first theorem - that minimax problem~\eqref{eq:feedback_law} is guaranteed to have a solution $\forall~x \in \mathcal{X}$.
\begin{theorem}
\label{thm:existence}
Let Assumption~\ref{assump:continous_assumption} hold.  The test synthesizer in~\eqref{eq:feedback_law} is guaranteed to have a solution $d \in \mathcal{D}$ for every $x\in \mathcal{X}$, \textit{i.e.}
\begin{equation}
    \forall~x\in\mathcal{X}~\exists~d\in\mathcal{D}\suchthat d = \testsynth(x).
\end{equation}
\end{theorem}
\noindent With respect to the second aspect of our problem statement then, we can define an instantaneous, controller-agnostic difficulty measure $M: \mathcal{X} \times \mathcal{D} \to \mathbb{R}$ as the interior maximization problem in~\eqref{eq:feedback_law}.
\begin{align}
M(x,d) & = \max_{u \in \mathcal{U}}~\mathcal{F}\left(\dot{h}^F(x,d,u), u,  \mathcal{U}(x,d), m\right). \label{eq:difficulty_measure}
\end{align}
As in the case of defining a set of feasible inputs, we need to show that this difficulty measure satisfies the conditions in Definition~\ref{def:difficulty_metric} to be a valid difficulty measure.  This leads to the following Lemma.
\begin{lemma}
\label{lem:continuous_valid_difficulty_measure}
$M$ as defined in equation~\eqref{eq:difficulty_measure} is a difficulty measure as per Definition~\ref{def:difficulty_metric} with feasible input space $\mathcal{U}(x,d)$ as per equation~\eqref{eq:feasible_input} and action discriminator $v$ as per equation~\eqref{eq:cont_act_discrim}.
\end{lemma}
\noindent Then, Theorem~\ref{thm:existence} and Lemma~\ref{lem:continuous_valid_difficulty_measure} directly provide for the following corollary regarding minimization of $M$.
\begin{corollary}
\label{corr:optimality}
Let Assumption~\ref{assump:continous_assumption} hold. The test synthesizer in~\eqref{eq:feedback_law} minimizes the difficulty measure $M$ in~\eqref{eq:difficulty_measure} over all $d \in \mathcal{D}$, \textit{i.e.}
\begin{align}
    \testsynth(x) & = \argmin_{d \in \mathcal{D}}~M(x,d) 
\end{align}
\end{corollary}

\subsection{Proof of Continuous-Time Results}
\label{sec:cont_proofs}
This section will contain all necessary lemmas and proofs for both Theorems stated in Section~\ref{sec:cont_main_results}.  To start, we will reiterate a known result from the study of minimax problems as taken from the proof for Theorem 1 in~\cite{fan1953minimax}:
\begin{lemma}(From Theorem 1 in~\cite{fan1953minimax})
\label{lem:minimax_solution}
Let $\mathbb{X}$ and $\mathbb{Y}$ be compact sets, and let $f: \mathbb{X} \times \mathbb{Y} \to \mathbb{R}$ be a function that is continuous in both its arguments.  The following minimax problem has a solution, \textit{i.e.}
\begin{equation}
    \exists~x^* \in \mathbb{X},~y^*\in\mathbb{Y} \suchthat f(x^*,y^*) = \min_{x \in \mathbb{X}}~\max_{y \in \mathbb{Y}}~f(x,y).
\end{equation}
\end{lemma}
\noindent Second, for any state $x$, we can partition the space of feasible tests $\mathcal{D}$ into a set that does not permit feasible inputs and its complement:
\begin{gather}
    \label{eq:gamma_set}
    \Gamma(x) \triangleq \{ d \in \mathcal{D}~|~\mathcal{U}(x,d) = \varnothing \}.
\end{gather}
With Lemma~\ref{lem:minimax_solution} and $\Gamma$ above, we can prove Theorem~\ref{thm:existence}.

\begin{proof}
This proof will follow a case by case argument.  These cases are (Case 1) $\Gamma(x) \neq \varnothing$ and (Case 2) $\Gamma(x) = \varnothing$.  Here, $\Gamma(x)$ is as defined in~\eqref{eq:gamma_set}.

\sectionspacing
\newidea{Case 1 $\Gamma(x) \neq \varnothing$:} In this case, $\forall~d\in\Gamma(x)$, $\mathcal{U}(x,d) = \varnothing$.  By definition of $\mathcal{F}$ in equation~\eqref{eq:filtering}, this implies
\begin{equation}
    \label{eq:filter_when_empty}
    \mathcal{F}\left(\dot h^F(x,d,u), u, \mathcal{U}(x,d), m\right) = m,~\forall~u \in \mathcal{U},~d \in \Gamma(x).
\end{equation}
with $m$ as in equation~\eqref{eq:min_satisfaction}.  As mentioned earlier, we still need to formally prove that such an $m$ exists.  This arises through simple application of the Extreme Value Theorem.  As $h^F \in C^1(\mathcal{X} \times \mathcal{D})$ by Assumption~\ref{assump:continous_assumption} and the dynamics in ~\eqref{eq:nom_sys} are control-affine, $\dot h^F(x,d,u)$ is continuous in all three of its arguments.  By compactness of the feasible spaces for the minimization problem in~\eqref{eq:min_satisfaction}, Extreme Value Theorem guarantees a solution to the same minimization problem and setting $m$ to be that solution suffices to prove existence of $m$.

For the remainder of the proof, we will use the notation offered by $M$ in equation~\eqref{eq:difficulty_measure} to represent the value of the inner maximization problem in equation~\eqref{eq:feedback_law}.  To prove the required result then, we need to show that $M(x,d)$ is lower bounded by some value, and that there exists some $d \in \mathcal{D}$ that achieves this value.  To start, we claim that $M(x,d) \geq m,~\forall~d \in \mathcal{D}$.  This is easily verifiable for all $d \in \Gamma(x)$ by equality~\eqref{eq:filter_when_empty}.  It remains to show this lower bound works for all $d \in \mathcal{D}$ such that $d \not \in \Gamma(x)$.  This stems via definition of $m$ in equation~\eqref{eq:min_satisfaction}.  For each $d \not \in \Gamma(x)$, $\mathcal{U}(x,d) \neq \varnothing$.  As a result,
\begin{equation}
    M(x,d) = \max_{u \in \mathcal{U}(x,d)}~\dot{h}^F(x,d,u) \geq m, ~\forall~d \in \mathcal{D} \cap \Gamma(x)^C.
\end{equation}
This concludes proving that $M(x,d) \geq m~\forall~d \in \mathcal{D}$.  To finish the proof, it requires that at least one test parameter vector $d \in \mathcal{D}$ ensures $M(x,d) = m$, and any vector $d \in \Gamma(x)$ satisfies this criteria, concluding the proof for this case.

\sectionspacing
\newidea{Case 2 $\Gamma(x) = \varnothing$:}
In this case, we note that the inner maximization problem in the feedback law~\eqref{eq:feedback_law} is equivalent to a Linear Program:
\begin{align}
        \min_{d \in \mathcal{D}} & \quad \max_{u \in \mathbb{R}^m} & & \hspace{-0.1 in}c(d)^T u, \label{eq:LP_inner}\\
        & \mathrm{subject~to~} & & \hspace{-0.1 in} Au \leq b, & & \hspace{-0.1 in}(\equiv u \in \mathcal{U}),\\
        & & & \hspace{-0.1 in} C(d) u \leq k(d), & & \hspace{-0.1 in} (\equiv u \in \mathcal{U}(x,d)).
\end{align}
LP duality turns equation~\eqref{eq:LP_inner} into the following:
\begin{gather}
    \min_{d \in \mathcal{D}, \lambda \geq 0, \mu \geq 0}~\max_{u \in \mathbb{R}^m} 
    \begin{bmatrix}
    c(d) \\
    -\lambda \\
    -\mu
    \end{bmatrix}^T
    \begin{bmatrix}
    u \\
    Au - b \\
    C(d)u - k(d)
    \end{bmatrix}. \label{eq:LP_reform}
\end{gather}
For minimax problem~\eqref{eq:LP_reform}, we note that if we further constrain the inner maximization problem such that $u \in \mathcal{U}$, this does not change the solution:
\begin{align}
    & \min_{d \in \mathcal{D}, \lambda \geq 0, \mu \geq 0}~\max_{u \in \mathcal{U}} ~
    \begin{bmatrix}
    c(d) \\
    -\lambda \\
    -\mu
    \end{bmatrix}^T
    \begin{bmatrix}
    u \\
    Au - b \\
    C(d)u - k(d)
    \end{bmatrix} \label{eq:LP_ucompact}\\
    & = \min_{d \in \mathcal{D}, \lambda \geq 0, \mu \geq 0, \gamma \geq 0}~\max_{u \in \mathbb{R}^m}~
    \begin{bmatrix}
    c(d) \\
    -\lambda \\
    -\mu \\
    -\gamma \\
    \end{bmatrix}^T
    \begin{bmatrix}
    u \\
    Au - b \\
    C(d)u - k(d) \\
    Au - b
    \end{bmatrix}, \\
    & = \min_{d \in \mathcal{D}, \beta \geq 0, \mu \geq 0}~\max_{u \in \mathbb{R}^m}~
    \begin{bmatrix}
    c(d) \\
    -\beta \\
    -\mu
    \end{bmatrix}^T
    \begin{bmatrix}
    u \\
    Au - b \\
    C(d)u - k(d)
    \end{bmatrix}, \\
    & = \eqref{eq:LP_reform}.
\end{align}

Additionally, for any $d \in \mathcal{D}$, minimax problem \eqref{eq:LP_reform} has a solution.  This stems from the fact that $\Gamma(x) = \varnothing$, and as a result, $\mathcal{U}(x,d)$ is a non-empty, closed polytope in $\mathbb{R}^m$ (see equation~\eqref{eq:gamma_set} for reference).  Therefore, the inner maximization problem in~\eqref{eq:LP_inner} has a solution, and via LP duality, so to does~\eqref{eq:LP_reform} have a solution.  Furthermore, as minimax problem~\eqref{eq:LP_ucompact} is equivalent to minimax problem~\eqref{eq:LP_reform}, so to does~\eqref{eq:LP_ucompact} have a solution for any $d \in \mathcal{D}$.  In addition, the $d$-dependent Lagrange multipliers for this solution $\lambda^*(d) < \infty$ and $\mu^*(d) < \infty$, as a solution $u^*$ exists.  As this is valid $\forall~d\in\mathcal{D}$, we note that $\exists~M_\lambda < \infty$ and $M_\mu < \infty$ such that $\lambda^*(d) \leq M_\lambda$ and $\mu^*(d) \leq M_\mu$ element-wise $\forall~d\in\mathcal{D}$.  If this were not the case, then there exists at least one $d \in \mathcal{D}$ such that $\lambda^*(d) \to \infty$ or $\mu^*(d) \to \infty$, implying infeasibility of the inner maximization problem in~\eqref{eq:LP_inner}, which is a contradiction.  As a result, we can uniformly upper bound $\lambda,\mu$ in~\eqref{eq:LP_ucompact} resulting in the following equality:
\begin{align}
    \hspace{-0.2 in}\eqref{eq:LP_inner} = 
    \min_{\substack{d \in \mathcal{D}, \\ 0 \leq \lambda \leq M_\lambda, \\ 0 \leq \mu \leq M_\mu}}~\max_{u \in \mathcal{U}}~
    \begin{bmatrix}
    c(d) \\
    -\lambda \\
    -\mu
    \end{bmatrix}^T
    \begin{bmatrix}
    u \\
    Au - b \\
    C(d)u - k(d)
    \end{bmatrix}. \label{eq:compact_minimax}
\end{align}
Finally, the minimax problem~\eqref{eq:compact_minimax} satisfies the conditions for Lemma~\ref{lem:minimax_solution}, guaranteeing a solution, \textit{i.e.} $\exists~d \in \mathcal{D}$ such that $d = \testsynth(x)$, proving the result for this case.

\sectionspacing
For an arbitrary $x \in \mathcal{X}$, the cases above prove that $\exists~d\in\mathcal{D}$ such that $d = \testsynth(x)$~\eqref{eq:feedback_law}.  As the choice of $x$ was left arbitrary, this result is valid $\forall~x\in\mathcal{X}$, thereby completing the proof.
\end{proof}

This concludes the proof for Theorem~\ref{thm:existence}.  It remains to prove Lemma~\ref{lem:continuous_valid_difficulty_measure} to show that we do indeed minimize a valid difficulty measure with our proposed approach.  The proof is as follows.

\begin{proof}
This proof will proceed in a case-by-case fashion as had the proof for Theorem~\ref{thm:existence}.  These cases will be $\Gamma(x) \neq \varnothing$ (Case 1) and $\Gamma(x) = \varnothing$ (Case 2).  In both cases however, we know via Theorem~\ref{thm:existence} that there exists a minimizer $\testsynth(x)$ of our proposed state-based difficulty metric $M$.  Here, $\testsynth(x)$ is defined in equation~\eqref{eq:feedback_law} and $M$ is defined in equation~\eqref{eq:difficulty_measure}.

\sectionspacing
\newidea{Case 1} $\Gamma(x) \neq \varnothing$: In this case, we know via the proof for Theorem~\ref{thm:existence} that the minimizer $\testsynth(x) \in \Gamma(x)$.  As a result, $\mathcal{U}(x,\testsynth(x)) = \varnothing$.  As such, we know that in this case, any minimizers $d$ of $M$ - which are guaranteed to exist via Theorem~\ref{thm:existence} - are such that $\mathcal{U}(x,d) = \varnothing$ indicating that they satisfy the first criteria for $M$ to be a valid difficulty measure as per Definition~\ref{def:difficulty_metric}.

\sectionspacing
\newidea{Case 2} $\Gamma(x) = \varnothing$: In this case, we know via the proof for Theorem~\ref{thm:existence} that the minimizer $\testsynth(x)$ yields a non-empty feasible input space, \textit{i.e.} $\mathcal{U}(x,\testsynth(x)) \neq \varnothing$.  As a result, by definition of $\testsynth(x)$ in equation~\eqref{eq:feedback_law}, the filtration function $\mathcal{F}$ in equation~\eqref{eq:filter_when_empty}, and the action discriminator $v$ in equation~\eqref{eq:cont_act_discrim} we have the following equality:
\begin{equation}
    \testsynth(x) = \argmin_{d \in \mathcal{D}}~\max_{u \in \mathcal{U}(x,d)}~v(x,d,u).
\end{equation}
As such, the minimizer $\testsynth(x)$ satisfies the second condition for $M$ to be a valid difficulty metric in this case.

As the choice of $x \in \mathcal{X}$ was left arbitrary, the prior logic holds $\forall~x \in \mathcal{X}$ thus concluding the proof that $M$ as per equation~\eqref{eq:difficulty_measure} is a valid difficulty measure as per Definition~\ref{def:difficulty_metric}.
\end{proof}  

Finally, it remains to prove Corollary~\ref{corr:optimality}, which is a direct consequence of Theorem~\ref{thm:existence} and Lemma~\ref{lem:continuous_valid_difficulty_measure}.

\begin{proof}
This is a consequence of Theorem~\ref{thm:existence} and Lemma~\ref{lem:continuous_valid_difficulty_measure}.
\end{proof}

\begin{figure*}[t]
    \centering
    \includegraphics[width = 0.99\textwidth]{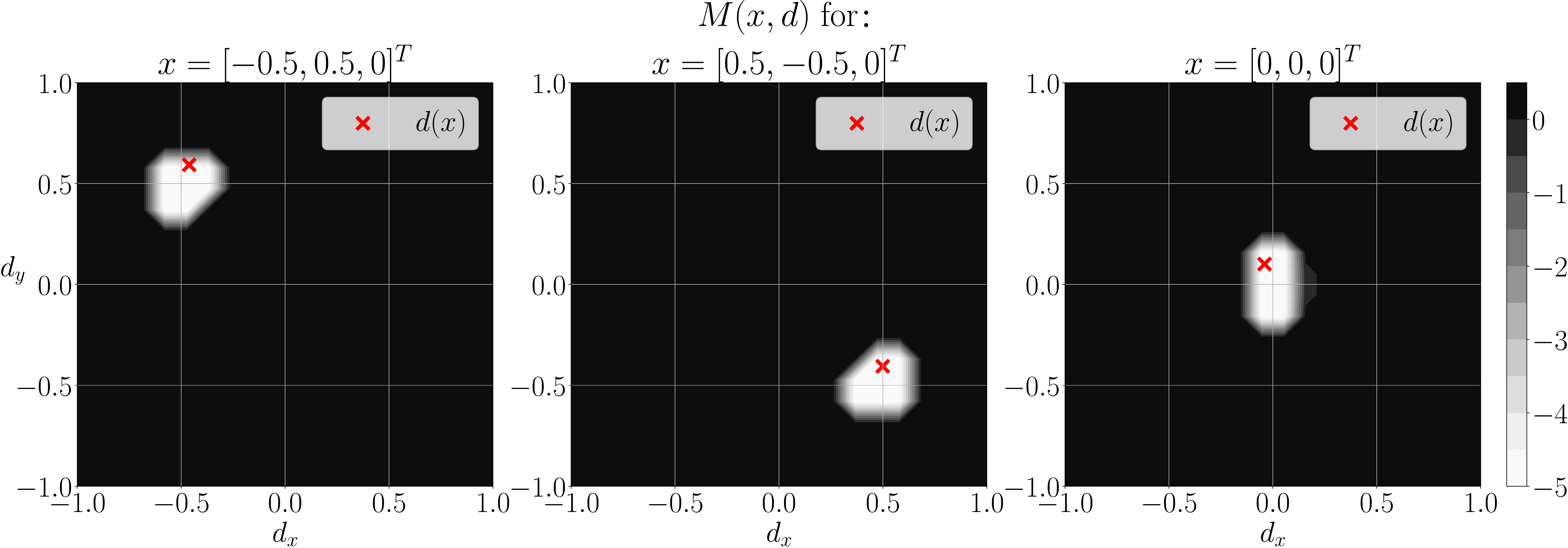}
    \caption{Minimization of the difficulty measure $M$ defined in equation~\eqref{eq:cont_ex_difficulty_measure}, for the autonomous agent example in Section~\ref{sec:cont_examples}. Notice that in each of the three cases shown, the output of the minimax test synthesizer in equation~\eqref{eq:cont_ex_test_synthesizer} accurately identifies a test that minimizes the corresponding difficulty measure - the colorbar is shown to the right hand side.}
    \label{fig:difficulty_measure}
    \vspace{-0.1 in}
\end{figure*}

\subsection{Corollaries - Perturbing the System Dynamics}
\label{sec:cont_corollaries}
In the prior section, we stated and proved two theorems regarding the existence and maximal difficulty of the adversarial, time-varying tests generated by our proposed technique, minimax problem~\eqref{eq:feedback_law}.  However, this setting only accounts for the scenario where the test permits perturbation of the truth regions for the timed reach-avoid predicates $\mu$ and $\omega_j$ (by Assumption~\ref{assump:cbf_stl}).  What if instead, we wanted a test where we simulated a motor failure, increased/decreased system friction, or other system-specific failures.  In this setting, the nominal control system changes to the following, where the dependence on the dynamics $f,g$ on $d$ correspond to simulations of motor failure and the $Cd$ term corresponds to increased/decreased friction:
\begin{equation}
    \label{eq:perturbed_sys}
    \dot x = f(x,d) + g(x,d)u + Cd,~C \in \mathbb{R}^{n \times p}.
\end{equation}
Here, we can still prove that our proposed test-synthesis procedure is still guaranteed to produce realizable and maximally difficult tests, as expressed through the following two Corollaries.  Corollary~\ref{cor:cont_existence_perturbed} states that our synthesizer in~\eqref{eq:feedback_law} is still guaranteed to produce realizable tests of system behavior in this setting.  Likewise, Corollary~\ref{cor:cont_optimality_perturbed} proves that these tests are maximally difficult with respect to the same difficulty metric $M$ as in~\eqref{eq:difficulty_measure}.  In both cases, all time derivatives are taken with respect to the dynamics in equation~\eqref{eq:perturbed_sys}.  We will start first with Corollary~\ref{cor:cont_existence_perturbed} which proves existence.
\begin{corollary}
\label{cor:cont_existence_perturbed}
Let the system dynamics be as in~\eqref{eq:perturbed_sys}, let Assumption~\ref{assump:continous_assumption} hold, and let $f$ and $g$ both be continuous in $d$.  The test synthesizer in~\eqref{eq:feedback_law} is guaranteed to have a solution $d \in \mathcal{D}$ $\forall~x \in \mathcal{X}$, \textit{i.e.}
\begin{equation}
    \forall~x \in \mathcal{X}~\exists~d \in \mathcal{D} \suchthat d = \testsynth(x).
\end{equation}
\end{corollary}
\begin{proof}
The proof for this corollary follows in the footsteps of the proof for Theorem~\ref{thm:existence} in a similar case-by-case fashion.  We can partition the feasible test-space $\mathcal{D}$ with $\Gamma(x)$ as defined prior in equation~\eqref{eq:gamma_set}, and set up the same two cases as prior.  In the first case, $\Gamma(x) \neq \varnothing$, changing the system dynamics does not change the outcome.  The optimal solution $\testsynth(x) \in \Gamma(x)$.  However, the latter case where $\Gamma(x) = \varnothing$ does change slightly.  The interior maximization problem is still a Linear Program, and the entire chain of logic until equation~\eqref{eq:compact_minimax} still holds.  However, to proceed with the last step and use Lemma~\ref{lem:minimax_solution} to complete the proof, we need to guarantee that the matrix multiplication in equation~\eqref{eq:compact_minimax} is continuous in $d, \lambda, \mu$ and $u$.  Continuity in $\lambda, \mu,$ and $u$ is assured via linearity in those terms.  Finally, continuity in $d$ is assured via the assumptions of continuity in the statement of Corollary~\ref{cor:cont_existence_perturbed}.  As a result, we can use Lemma~\ref{lem:minimax_solution}, thus completing the proof.
\end{proof}

\noindent Next, Corollary~\ref{cor:cont_optimality_perturbed} will prove the optimal difficulty of the tests generated via our synthesizer in the perturbed setting.

\begin{corollary}
\label{cor:cont_optimality_perturbed}
Let the system dynamics be as in~\eqref{eq:perturbed_sys}, let Assumption~\ref{assump:continous_assumption} hold, and let $f$ and $g$ all be continuous in $d$.  The test synthesizer in~\eqref{eq:feedback_law} minimizes the difficulty measure $M$ in equation~\eqref{eq:difficulty_measure}, \textit{i.e.},
\begin{equation}
    \testsynth(x) = \argmin_{d \in \mathcal{D}}~M(x,d).
\end{equation}
\end{corollary}
\begin{proof}
This is a consequence of Corollary~\ref{cor:cont_existence_perturbed}.
\end{proof} 


\subsection{Examples}
\label{sec:cont_examples}
In this section we will illustrate our main results through examples extending Examples~\ref{ex:turtle} and~\ref{ex:turtle_2}.  For completeness, we will state the system dynamics as follows:
\begin{equation}
\label{eq:turtle_dynamics}
\begin{gathered}
    x = \begin{bmatrix}
    x \\
    y \\
    \theta
    \end{bmatrix},~
    \dot x =
    \underbrace{\begin{bmatrix}
    \cos{\theta} & 0 \\
    \sin{\theta} & 0 \\
    0 & 1
    \end{bmatrix}}_{g(x)} u, ~
    x \in [-1,1]^2 \times [0, 2\pi],
    \\
    u = [u_1,u_2]^T \in [-0.2, 0.2] \times [-1,1].
\end{gathered}
\end{equation}
Equation~\eqref{eq:turtle_dynamics} implies that $\mathcal{X} = [0,1]^2 \times [0, 2\pi]$ and $\mathcal{U} = [-0.2,0.2] \times [-1,1]$, and both satisfy the conditions for Assumption~\ref{assump:continous_assumption}.  To generate a minimax test synthesizer of the form in equation~\eqref{eq:feedback_law} we require a system specification and associated control barrier functions.  For our example, our specification,
\begin{equation}
    \label{eq:example_spec}
    \begin{gathered}
        \psi = \F \mu \wedge_{j \in \mathcal{J}} \G \omega_j, \\
        \llbracket \mu \rrbracket = \left\{x\in\mathcal{X}~|~\|Px - g\| \leq 0.25 \right\}, \\ \llbracket \omega_j \rrbracket = \left\{x \in \mathcal{X} ~|~ \|Px - o_j\| \geq 0.175 \right\}.
    \end{gathered}
\end{equation}
\begin{figure*}[t]
    \centering
    \includegraphics[width = 0.99\textwidth]{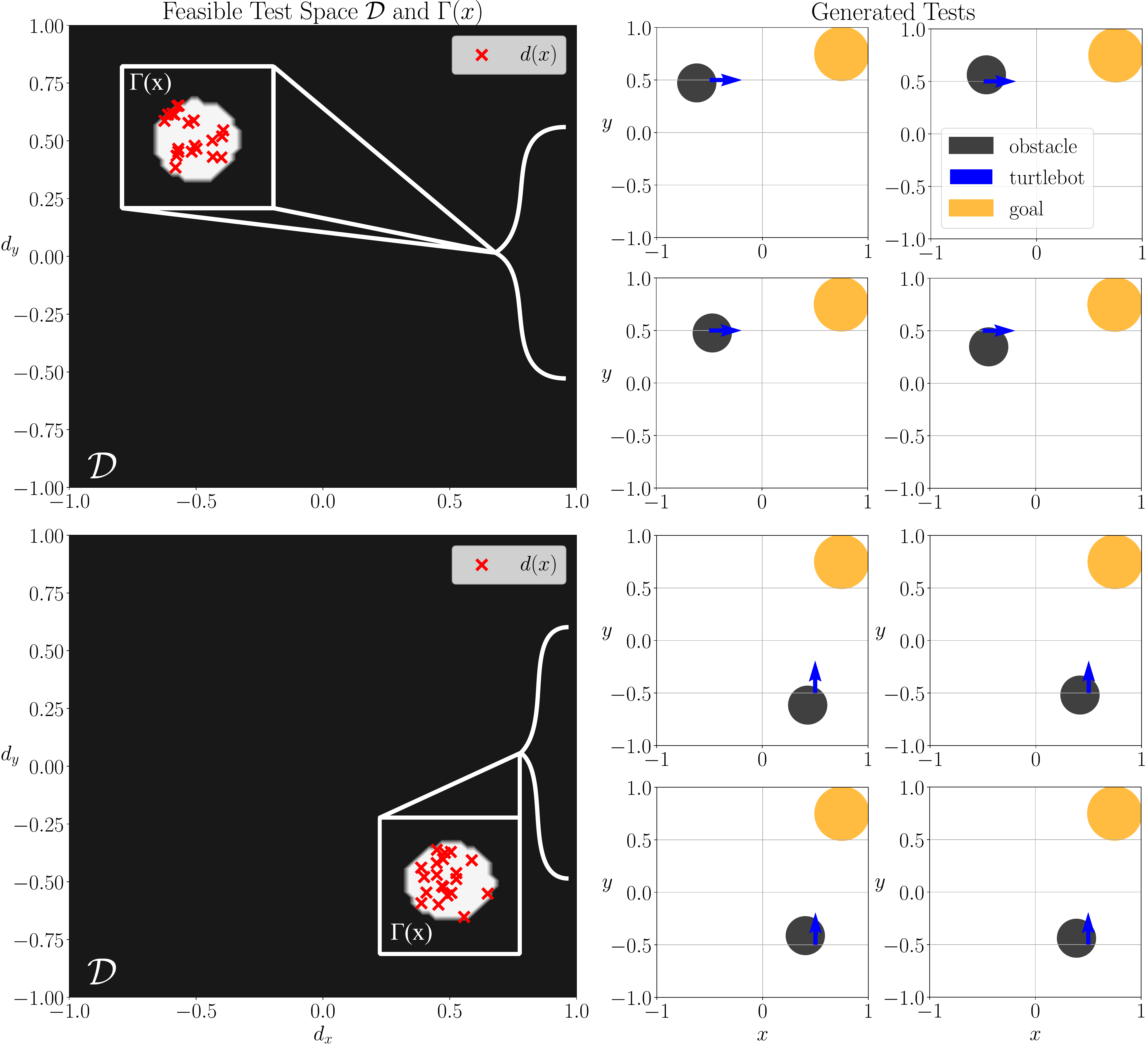}
    \caption{Example obstacle placements produced by our test-synthesis procedure defined in equation~\eqref{eq:cont_ex_test_synthesizer} for the unicycle example in Section~\ref{sec:cont_examples}.
    (Top Left) Feasible Test Space $\mathcal{D}$ with the partitioning offered by $\Gamma(x)$ for $x = [-0.5,0.5,\pi/4]^T$.  $\Gamma(x)$ is defined in equation~\eqref{eq:gamma_set}.  Notice how twenty different solutions to the same minimax problem~\eqref{eq:cont_ex_test_synthesizer} all yield a test parameter vector $d \in \Gamma(x)$ as theorized in the proof for Theorem~\ref{thm:existence}.  (Top Right) Four specific tests generated by the same test-synthesizer in equation~\eqref{eq:cont_ex_test_synthesizer} for the same system state.  (Bottom) Figures showing the same information as above, except the system state has changed to $x = [0.5,-0.5,\pi/2]^T$.}
    \label{fig:cont_example_solutions}
    \vspace{-0.2 in}
\end{figure*}
\indent Here, $g \in \mathbb{R}^2$ denotes the center of a goal region the system is to enter, $o_j \in \mathbb{R}^2$ denotes the center of an obstacle the system is to stay away from, and $P$ projects the system state onto the $x-y$ plane.  Then, the control barrier functions for $\psi$ are,
\begin{equation}
    \label{eq:example_cbfs}
    \begin{gathered}
    h^F(x) = 0.25^2 - \|Px - g\|_2^2, \\
    h^G_j(x,o_j) = \|Px - o_j\|_2^2 - 0.175^2,\\
    \frac{\partial h^F}{\partial x} = -2P^T(Px-g),~\frac{\partial h^G_j}{\partial x} = 2P^T(Px-o_j).
    \end{gathered}
\end{equation}
Additionally, it is easily verifiable that both control barrier functions above satisfy the conditions for Assumptions~\ref{assump:cbf_stl} and~\ref{assump:continous_assumption}.  In this case, $h^F(x)\geq 0 \iff x \in \llbracket \mu \rrbracket$ and $h^G_j(x,o_j) \geq 0 \iff x \in \llbracket \omega_j \rrbracket$.  Finally, we need to formalize our test parameter vector $d$ and the space in which it lives.  In continuing with Example~\ref{ex:turtle_2}, we will assume the only perturbable objects in the environment are the center locations of our obstacles $o_j$.  We will assume the number of obstacles $N_o = |\mathcal{J}|$, resulting in the following test parameter vector:
\begin{equation}
    \label{eq:example_test_setup}
    d = [o_1^T, o_2^T, \dots]^T, \quad d \in \mathcal{D} = [-1,1]^{2N_o} \subset \mathbb{R}^{2N_o}.
\end{equation}

For this example setting, our minimax testing law $\testsynth: \mathcal{X} \to \mathcal{D}$, feasible input space $\mathcal{U}(x,d)$, and difficulty measure $M$ are as follows, with $\mathcal{J} = \{1\}$ as we will show examples with only one obstacle:
\begin{align}
    \label{eq:example_feasible_space}
    \mathcal{U}(x,d) & = \left\{u \in \mathcal{U}~\Bigg|~\frac{\partial h^G_1}{\partial x}^T u \geq -10 h^G_1(x,o_1)\right\}, \\
    M(x,d) & = \max_{u \in \mathcal{U}}~\mathcal{F}\left(\frac{\partial h^F}{\partial x}^T u, u , \mathcal{U}(x,d), -5\right), \label{eq:cont_ex_difficulty_measure} \\
    \testsynth(x) & = \argmin_{d \in \mathcal{D}}\max_{u \in \mathcal{U}}~\mathcal{F}\left(\frac{\partial h^F}{\partial x}^T u, u , \mathcal{U}(x,d), -5\right). \label{eq:cont_ex_test_synthesizer}
\end{align}
Finally, we note that all optimization problems to be solved in this section will utilize a variant on the algorithm described in~\cite{akella2021learning}.

\sectionspacing
\newidea{Figure Analysis:}  It is evident that the example autonomous agent setting described in equation~\eqref{eq:turtle_dynamics} with corresponding control barrier functions in equation~\eqref{eq:example_cbfs} and $d$ defined in equation~\eqref{eq:example_test_setup} satisfies the conditions for Theorem~\ref{thm:existence} and Corollary~\ref{corr:optimality}.  As a result, we would expect that our test synthesizer $\testsynth$ defined in equation~\eqref{eq:cont_ex_test_synthesizer} should always produce a realizable test of system behavior $\forall~x \in \mathcal{X}$.  We also expect this generated test to minimize the difficulty measure $M$ defined in equation~\eqref{eq:cont_ex_difficulty_measure}.  Figure~\ref{fig:difficulty_measure} shows this result for three separate system states $x = [-0.5,0.5,0]^T$, $x=[0.5,-0.5,0]^T$, and $x=[0,0,0]^T$.  Specifically, notice how the red "x" indicating the solution to the test synthesizer~\eqref{eq:cont_ex_test_synthesizer} always lies within the minimizing region of $M(x,d)$ at that state $x$ - the white region in each graph.  Indeed, over $1000$ randomized runs where the initial state $x$ is perturbed uniformly over the state-space $\mathcal{X}$ in equation~\eqref{eq:turtle_dynamics}, the test synthesizer finds a test parameter $d$ such that $M(x,d) = -5$ - its minimum value.
\indent Figure~\ref{fig:cont_example_solutions} goes a step farther and shows $8$ example tests generated by our test synthesizer.  As shown in the left hand side figures in Figure~\ref{fig:cont_example_solutions}, each of the twenty solutions to the test synthesizer in~\eqref{eq:cont_ex_test_synthesizer} lie in the partitioning set $\Gamma(x)$ within the feasible test space $\mathcal{D}$.  Solutions are shown via red "x"-es and $\Gamma(x)$ via the white regions in both left-hand side figures.  This phenomenon of the solutions lying within the non-empty set $\Gamma(x)$ is expected via the proof for Theorem~\ref{thm:existence} as $\Gamma(x) \neq \varnothing$.  Specifically, the reason $\Gamma(x) \neq \varnothing$ is that we have not constrained against the obstacle lying on top of the agent to be tested.  Due to the agent's limited actuation capacity, placing the obstacle on top of the agent results in an infeasible control barrier function condition, rendering $\mathcal{U}(x,d) = \varnothing$.  This fact is corroborated through the $8$ example tests shown on the right-hand side.  In each of these tests, the obstacle lies atop the agent which is located at the base of the arrow and is heading in the direction the arrow indicates.  Ideally, we would like to constrain against such trivial solutions in our test-synthesis framework, and we do so in Section~\ref{sec:extensions}.

\section{Discrete-Time Test Generation}
\label{sec:discrete}
Similar to the prior section, this section will state and prove the latter half of our main results - the development of an adversarial, time-varying test-synthesis procedure for discrete-time control systems of the form in~\eqref{eq:discrete_system} subject to timed reach-avoid specifications $\psi$ as in~\eqref{eq:spec}.  As before, we will briefly describe the overarching methodology behind our approach, state the developed minimax problem for test synthesis, and end with two, similar theorems to the continuous case.

\sectionspacing
\newidea{Overarching Idea:}  As in the continuous case, Assumption~\ref{assump:cbf_stl} lets us express satisfaction of the reach-avoid specification $\psi$~\eqref{eq:spec} via control barrier functions and their $0$-superlevel sets.  Specifically, $\forall~k \in \mathbb{Z}_+$ and $k_{\max} = \min\{k\in\mathbb{Z}_+~|~t_{\max} \leq k \Delta t\}$ for some $\Delta t \geq 0$, the discrete state trajectory $x_k$ $\forall~k\in\mathbb{Z}_+$ satifies $\psi$ at $k=0$, \textit{i.e.} $(x,0) \models \psi$ if and only if:
\begin{equation}
    x_k \in \cap_{j \in \mathcal{J}}~\mathcal{C}_{h^G_j}~\mathrm{and~}\exists~k \in \{0,1,\dots,k_{\max}\} \suchthat x_k \in \mathcal{C}_{h^F}.
\end{equation}
\noindent As in the continuous setting, we will also make the same assumption on the starting, system state as expressed in equation~\eqref{eq:continuous_cbf_start}:
\begin{equation}
    x_0 \in \cap_{j \in \mathcal{J}}~\mathcal{C}_{h^G_j}~\mathrm{and~}x_0 \not \in \mathcal{C}_{h^F}.
\end{equation}
\noindent Here however lies a difference.  If we naively used the same control barrier function decrement conditions in Definition~\ref{def:discrete_cbf} to identify inequalities to constrain a minimax problem for test generation, the resulting inner maximization problem would not be concave.  As a result, we would not be able to use Theorem~\ref{thm:existence} and its proof to gain any insight into this scenario.  However, the reason we required concavity of the inner maximization problem was that concavity guaranteed a solution - a feasible control input for a given test parameter vector $d$.  To facilitate the provision of similar guarantees, we will make the following assumption.
\begin{assumption}
\label{assump:finite_spaces}
The input space $\mathcal{U}$ for the discrete-time system~\eqref{eq:discrete_system} and the feasible test space $\mathcal{D}$ are finite, \textit{i.e.} $|\mathcal{U}| < \infty$ and $|\mathcal{D}| < \infty$.
\end{assumption}
\noindent For context, this assumption is easily satisfied by any system described by a finite-action Markov Decision Process with the space of feasible tests corresponding to edges that can be turned on/off.

In this setting, we can still define a space of feasible inputs $\mathcal{U}(x,d)$ and an action discriminator $v$ as we did for the continuous setting.
\begin{align}
    \label{eq:disc_feasible_input}
    \mathcal{U}(x,d) & = \left\{u \in \mathcal {U}~\bigg|~\forall~j\in\mathcal{J}~h^G_j(f(x,u),d) \geq 0 \right\}, \\
    \label{eq:disc_act_discrim}
    v(x,d,u) & = h^F(f(x,u),d) - h^F(x,d) - \tau,~\tau > 0.
\end{align}
In effect then, our proposed test synthesizer will be very similar to its continuous-time counterpart.  Specifically, we will still have an outer minimization problem over a space of feasible tests.  Additionally, the goal is to minimize the maximum possible increment in a control barrier function subject to the enduring positivity of multiple other control barrier functions.  Keeping these parallels in mind, the statement and main results for our proposed, adversarial, time-varying discrete-time test-synthesis procedure will follow.

\subsection{Statement of Discrete-Time Results}
\label{sec:disc_main_contribution}
In order to provide a parallel to the continuous setting, we will first define a difference function for the incremental change in a control barrier function after an action has been taken.
\begin{equation}
    \label{eq:difference_func}
    \Delta h(x,u,d) = h\left(f(x,u), d \right) - h(x,d).
\end{equation}
Then, our proposed test-generation method is as follows:
\begin{align}
    \label{eq:disc_feedback_law}
    \testsynth(x) & = \argmin_{d \in \mathcal{D}}\max_{u \in \mathcal{U}}~\mathcal{F}\left(\Delta h^F(x,u,d), u, \mathcal{U}(x,d),m\right).
\end{align}
Here, $\mathcal{F}$ is defined in equation~\eqref{eq:filter_when_empty}.  In the discrete setting, the definition of $m$ changes slightly and will be reproduced here:
\begin{equation}
    \label{eq:discrete_min_satisfaction}
    m \leq \min_{u \in \mathcal{U},~x \in \mathcal{X},~d \in \mathcal{D}}~\Delta h^F\left(x,u,d\right).
\end{equation}
As before, we have not formally stated whether such an $m$ exists.  However, we will prove its existence in the proofs to follow.  Intuitively though, $m$ is defined as the lower bound to a finite series of finite optimization problems, each of which is guaranteed to have a solution.  Therefore, so too is $m$ guaranteed to exist.
As before, this results in our second theorem which states that minimax problem~\eqref{eq:disc_feedback_law} is guaranteed to have a solution $\forall~x\in\mathcal{X}$.
\begin{theorem}
\label{thm:discrete_existence}
Let Assumption~\ref{assump:finite_spaces} hold.  The test synthesizer in~\eqref{eq:disc_feedback_law} is guaranteed to have a solution $d \in \mathcal{D}$ for every $x\in\mathcal{X}$:
\begin{equation}
    \mathrm{\textit{i.e.,~}}\forall~x\in\mathcal{X}~\exists~d\in\mathcal{D}\suchthat d=\testsynth(x).
\end{equation}
\end{theorem}
\noindent Additionally, we can also define a very similar difficulty measure $\bar M$ as to its continuous counterpart $M$ as in equation~\eqref{eq:difficulty_measure}:
\begin{equation}
\bar M(x,d) = \max_{u \in \mathcal{U}}~\mathcal{F}\left(\Delta h^F(x,u,d), u, \mathcal{U}(x,d),m\right). \label{eq:discrete_difficulty_measure}
\end{equation}
As before, we need to prove that our proposed difficulty measure satisfies Definition~\ref{def:difficulty_metric}.  The following lemma expresses this statement.
\begin{lemma}
\label{lem:discrete_valid_difficulty_measure}
$\bar M$ as defined in equation~\eqref{eq:discrete_difficulty_measure} is a valid difficulty measure as per Definition~\ref{def:difficulty_metric} with feasible input space $\mathcal{U}(x,d)$ as per equation~\eqref{eq:disc_feasible_input} and action discriminator $v$ as per equation~\eqref{eq:disc_act_discrim}.
\end{lemma}
\noindent Finally, with respect to this difficulty measure $\bar M$ we have another corollary regarding the optimal difficulty of the generated tests:
\begin{corollary}
\label{corr:discrete_optimality}
Let Assumption~\ref{assump:finite_spaces} hold.  The test synthesizer in~\eqref{eq:disc_feedback_law} minimizes the difficulty measure $\bar M$ in~\eqref{eq:discrete_difficulty_measure} over all $d \in \mathcal{D}$, \textit{i.e.}
\begin{equation}
    \testsynth(x) = \argmin_{d\in\mathcal{D}}~\bar M(x,d).
\end{equation}
\end{corollary}
\noindent As before, we will prove these statements in the next section.

\subsection{Proof of Discrete-Time Results}
\label{sec:disc_proofs}
Similar to Section~\ref{sec:cont_corollaries}, before stating the proofs of both main results in the discrete-time setting, we will first state a useful Lemma.
\begin{lemma}
\label{lem:finite_optimality}
For any non-empty, finite set $A \subset \mathbb{R}, ~\exists~m,M \in \mathbb{R}~\mathrm{s.t.}$,
\begin{equation}
    m \leq a \leq M,~\forall~a \in A.
\end{equation}
\end{lemma}
\noindent The proof of Theorem~\ref{thm:discrete_existence} will follow.

\begin{proof}
This proof amounts to two separate uses of Lemma~\ref{lem:finite_optimality} and will follow a similar partitioning analysis as in the continuous setting.  We will group the cases for expediency. Specifically, we can define a barred-Gamma set similar to its counterpart $\Gamma(x)$ as defined in equation~\eqref{eq:gamma_set} and with $\mathcal{U}(x,d)$ the feasible input set~\eqref{eq:disc_feasible_input}:
\begin{equation}
    \label{eq:bar_gamma}
    \bar{\Gamma}(x) = \{d \in \mathcal{D}~|~\mathcal{U}(x,d) = \varnothing \}.
\end{equation}
We will also use the notation offered by $\bar M$ in equation~\eqref{eq:discrete_difficulty_measure} to denote the value of the inner maximization problem in equation~\eqref{eq:disc_feedback_law}.
Then, in the discrete-setting we can rewrite minimax problem~\eqref{eq:disc_feedback_law} with $\bar M (x,d)$ as follows, based on the definition of $\mathcal{F}$ in equation~\eqref{eq:filter_when_empty}:
\begin{equation}
    \label{eq:resolved_discrete_opt}
    \testsynth(x) = \argmin_{d \in \mathcal{D}}~
    \begin{cases}
        \bar M(x,d) & \mbox{if}~d \not \in \bar{\Gamma}(x) \\
        m & \mbox{else}.
    \end{cases}
\end{equation}
Then, the two cases can be resolved simultaneously.  In the event that $\bar{\Gamma}(x) = \varnothing$, the above optimization problem collapses to a minimization of $\bar M(x,d)$.  Each $\bar M(x,d)$ is guaranteed to exist via Lemma~\ref{lem:finite_optimality}, and as a result, a solution to the larger optimization problem is guaranteed to exist via Lemma~\ref{lem:finite_optimality} as the space of all tests $\mathcal{D}$ is finite.  In the event that $\bar{\Gamma}(x) \neq \varnothing$, then any choice of $d \in \bar{\Gamma}(x)$ yields $m \leq \bar M(x,d')~\forall~d' \in \mathcal{D} \cap \bar{\Gamma}(x)^C$.  As such the choice of $d \in \bar{\Gamma}(x)$ solves the above optimization problem.  This holds $\forall~x \in \mathcal{X}$, thus concluding the proof.
\end{proof}

Likewise, the proof for Lemma~\ref{lem:discrete_valid_difficulty_measure} will follow.

\begin{proof}
We will follow a case-by-case analysis with the cases offered by $\bar{\Gamma}(x)$ defined in equation~\eqref{eq:bar_gamma}.  In either case however, as per Theorem~\ref{thm:discrete_existence}, we know that there exists a $d \in \mathcal{D}$ $\forall~x \in \mathcal{X}$ that solves minimax problem~\eqref{eq:disc_feedback_law}.  By definition of the proposed difficulty measure $\bar M$ in equation~\eqref{eq:discrete_difficulty_measure}, so to do these solutions also minimize $\bar M$.  To be clear in the remainder of this proof, we will call these solutions $d^*$ mimicking the notation used in Definition~\ref{def:difficulty_metric}.  Then, in the event that $\bar \Gamma(x) \neq \varnothing$, by the proof for Theorem~\ref{thm:discrete_existence} we know that $d^* \in \bar \Gamma(x)$.  As a result, $\mathcal{U}(x,d^*) = \varnothing$ by definition of $\bar \Gamma(x)$ in equation~\eqref{eq:bar_gamma}.  Therefore, $\bar M$ satisfies the first condition for being a difficulty measure.  In the second case, $\bar \Gamma(x) = \varnothing$ and by the proof for Theorem~\ref{thm:discrete_existence} and definition of the action discriminator $v$ in equation~\eqref{eq:disc_act_discrim} we have the following:
\begin{align}
    d^* & = \argmin_{d \in \mathcal{D}}~\max_{u \in \mathcal{U}(x,d)}~\Delta h^F(x,u,d), \\
    & = \argmin_{d \in \mathcal{D}}~\max_{u \in \mathcal{U}(x,d)}~v(x,d,u) + \tau.
\end{align}
Therefore, in the case where $\bar \Gamma(x) = \varnothing$, $\bar M$ satisfies the second condition to be a difficulty measure as per Definition~\ref{def:difficulty_metric}.  
\end{proof}

The proof of Corollary~\ref{corr:discrete_optimality} then stems from Theorem~\ref{thm:discrete_existence} and Lemma~\ref{lem:discrete_valid_difficulty_measure}.

\begin{proof}
This is a consequence of Theorem~\ref{thm:discrete_existence} and Lemma~\ref{lem:discrete_valid_difficulty_measure}.
\end{proof}

\subsection{Corollaries - Predictive Test Synthesis}
\label{sec:disc_extensions}
As mentioned, the discrete setting also permits us to predict future system states as well.  We will show that generating tests in this predictive framework amounts to a simple change in notation, with the majority of the prior section's analysis carrying over.  To start, we will assume an arbitrary $N$-step horizon for predictive test synthesis.  To do so requires a few definitions.  The first will provide a notational simplification for arbitrary, finite $N$-step horizon state predictions.
\begin{equation}
    \label{eq:discrete_predictive_dynamics}
    \begin{gathered}
    \mathbf{u} = [u_1, u_2, \dots, u_N], \quad \mathbf{f}(x,\mathbf{u},2) = f(f(x,u_1),u_2), \\
    x^N_{\mathbf{u}} = \mathbf{f}(x,\mathbf{u},N).
    \end{gathered}
\end{equation}
The next set of definitions formalizes the set of feasible input sequences in this predictive setting and identifies an action discriminator satisfying Definition~\ref{def:heading}.  Here, we note that $\mathbf{u}$ and $x^N_{\mathbf{u}}$ are as defined in equation~\eqref{eq:discrete_predictive_dynamics}, and $\mathcal{U}^N = \mathcal{U} \times \mathcal{U} \dots$ $N$ times.
\begin{align}
    \label{eq:disc_pred_feas_input}
    \hspace{-0.075 in}\mathcal{U}^N(x,d) & = \left\{ \mathbf{u} \in \mathcal{U}^N~\Bigg|~\forall~j\in\mathcal{J},~ h^G_j\left(x^N_{\mathbf{u}},d\right) \geq 0\right\}, \\
    \label{eq:predictive_difference}
    \hspace{-0.075 in} \Delta^N h(x,\mathbf{u},d) & = h\left(x^N_{\mathbf{u}}, d\right) - h(x,d), \\
    \label{eq:disc_pred_act_discrim}
    \hspace{-0.075 in} v(x,d,u ) & = \Delta^N h^F(x,\mathbf{u},d) - \tau.
\end{align}
With these terms, we can propose a predictive test-synthesis procedure that is similar to its one-step counterpart in~\eqref{eq:disc_feedback_law}.  Our proposed test-synthesis procedure and difficulty measure $\Tilde M$ are as follows:
\begin{align}
    \label{eq:disc_feedback_law_predictive}
    \testsynth(x) & = \argmin_{d \in \mathcal{D}}\max_{\mathbf{u} \in \mathcal{U}^N}~\xi^N(x,\mathbf{u},d), \\
    \Tilde M^N(x,d) & = \max_{\mathbf{u} \in \mathcal{U}^N}~\xi^N(x,\mathbf{u},d) \label{eq:discrete_difficulty_measure_predictive}, \\
    \xi^N(x,\mathbf{u},d) & = \mathcal{F}\left(\Delta^N h^F(x,\mathbf{u},d),\mathbf{u}, \mathcal{U}^N(x,d), m\right).
\end{align}
As prior, we define $m$ as follows:
\begin{equation}
    m \leq \min_{\mathbf{u} \in \mathcal{U}^N,~x \in \mathcal{X},~d\in\mathcal{D}}~\Delta^N h^F(x,\mathbf{u},d).
\end{equation}

With these definitions and equations, the next corollary formally states that the test-synthesis procedure in equation~\eqref{eq:disc_feedback_law_predictive} is guaranteed to produce realizable tests of system behavior.
\begin{corollary}
\label{cor:disc_predictive_existence}
Let Assumption~\ref{assump:finite_spaces} hold.  The test synthesizer in~\eqref{eq:disc_feedback_law_predictive} is guaranteed to have a solution $d \in \mathcal{D}~\forall~x \in \mathcal{X}$, \textit{i.e.},
\begin{equation}
    \forall~x \in \mathcal{X},~\exists~d \in \mathcal{D} \suchthat d = \testsynth(x).
\end{equation}
\end{corollary}
\begin{proof}
The proof for this corollary follows directly in the footsteps of the proof for Theorem~\ref{thm:discrete_existence}. More aptly, for any choice of finite prediction horizon $N$, we can make the following redefinition.
\begin{equation}
    x_{k+1} = x^N_{\mathbf{u}} = \mathbf{f}(x_k,\mathbf{u},N) = \Tilde{f}(x_k,\mathbf{u}).
\end{equation}
This redefinition effectively constructs a new, single-step discrete-time system whose input space $\mathcal{U}^N$ is still finite.  The result then stems from the direct application of Theorem~\ref{thm:discrete_existence}.
\end{proof}

\noindent In a similar fashion, we can also prove that the proposed difficulty measure $\bar{M}^N$ is a valid difficulty measure as per Definition~\ref{def:difficulty_metric}.
\begin{lemma}
$\Tilde M^N$ as defined in equation~\eqref{eq:discrete_difficulty_measure_predictive} is a valid difficulty measure as per Definition~\ref{def:difficulty_metric} with feasible input space $\mathcal{U}^N(x,d)$ as per equation~\eqref{eq:disc_pred_feas_input} and action discriminator $v$ as per equation~\eqref{eq:disc_pred_act_discrim}.
\end{lemma}
\begin{proof}
Following the same redefinition as in the proof for the prior Corollary, we find that our proposed difficulty measure $\bar{M}^N$ collapses to a one-step difficulty measure where the input space $\mathcal{U}^N$ is still finite.  The result then stems via application of Lemma~\ref{lem:discrete_valid_difficulty_measure}.
\end{proof}

\begin{figure}[t]
    \centering
    \hspace{-0.225 cm}\includegraphics[width =0.5\textwidth]{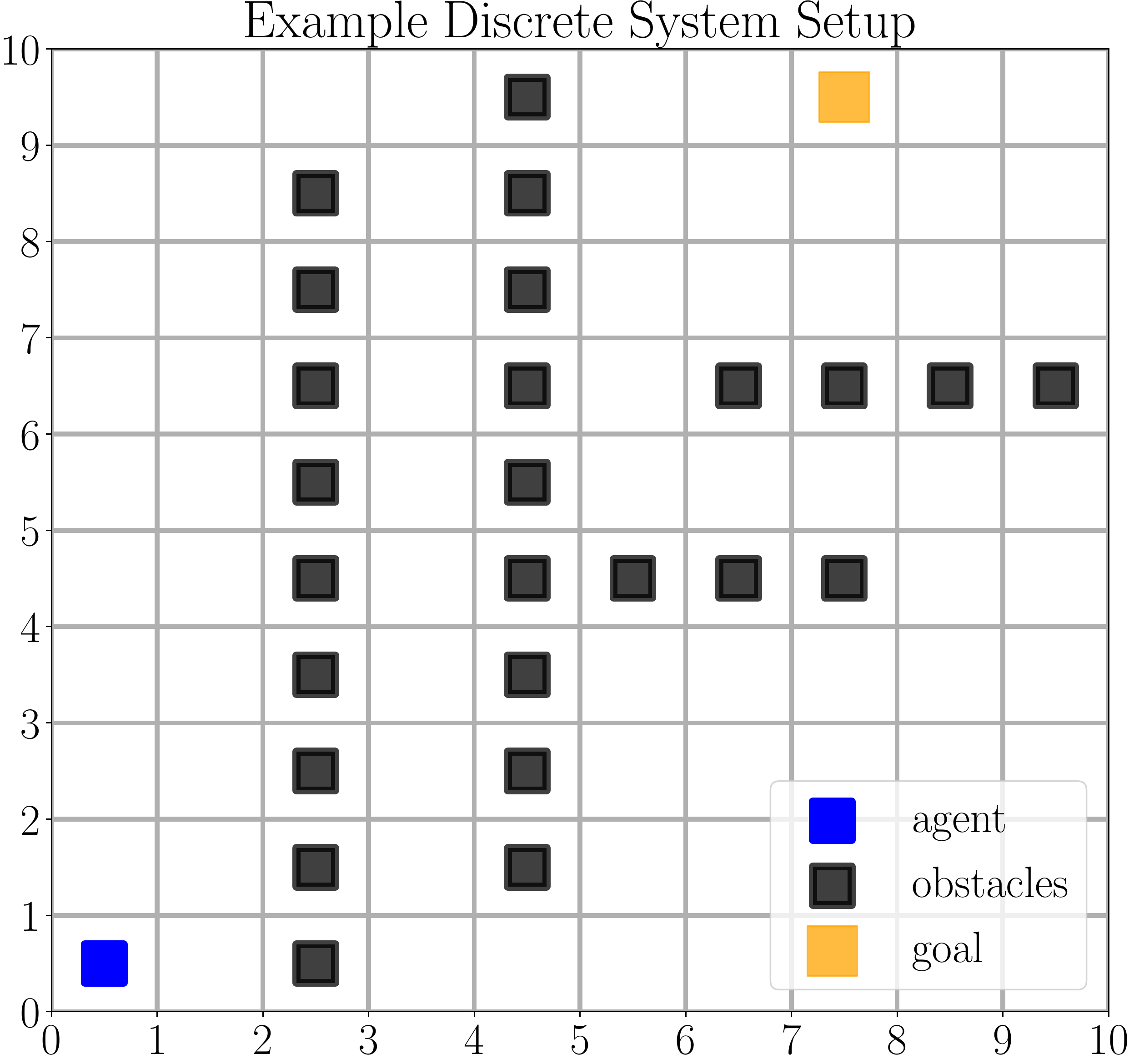}
    \caption{Example discrete time setting for Section~\ref{sec:disc_extensions}.}
    \label{fig:discrete_example}
    \vspace{-0.2 in}
\end{figure}

\noindent In a similar fashion, we can also formally state and prove that the tests generated by minimax problem~\eqref{eq:disc_feedback_law_predictive} are maximally difficult.
\begin{corollary}
\label{cor:disc_predictive_optimality}
Let Assumption~\ref{assump:finite_spaces} hold.  The test synthesizer in~\eqref{eq:disc_feedback_law_predictive} minimizes the difficulty measure $\Tilde M$ in~\eqref{eq:discrete_difficulty_measure_predictive} over all $d \in \mathcal{D}$, \textit{i.e.}
\begin{equation}
    \testsynth(x) = \argmin_{d \in \mathcal{D}}~\Tilde M^N(x,d).
\end{equation}
\end{corollary}
\begin{proof}
Again, this Corollary stems directly from Corollary~\ref{cor:disc_predictive_existence}.
\end{proof}

\subsection{Examples}
\label{sec:disc_examples}
Figure~\ref{fig:discrete_example} provides a picture for our example in this section - placement of obstacles on a discrete grid to frustrate an agent's ability to reach its goal.  The agent is modeled as a discrete transition system:
\begin{equation}
    \label{eq:discrete_example_system}
    \begin{gathered}
        x_{k+1} = \underbrace{ 
        \begin{cases}
            x_k \pm [1,0] & \mbox{if}~ u_k = \mathrm{left}~(-)~\mathrm{or}~\mathrm{right}~(+) \vspace{0.1 in}\\ 
            x_k \pm [0,1] & \mbox{if}~ u_k = \mathrm{down}~(-) ~\mathrm{or}~ \mathrm{up}~(+) \\
            x_k & \mbox{if}~u_k = \mathrm{stay} ~\mathrm{or}~ \mathrm{action~infeasible}.
        \end{cases}}_{f(x_k,u_k)}\\
        x_k \in \{0,1,2,\dots,9\}^2 = \mathcal{X}, \\
        u_k \in \{\mathrm{left},~\mathrm{right},~\mathrm{up},~\mathrm{down},~\mathrm{stay}\} = \mathcal{U}
    \end{gathered}
\end{equation}

\noindent Our test parameter $d$ and specification $\psi$ are as follows, with $g = [g^0, g^1] \in \mathcal{X}$ the goal-cell:
\begin{gather}
     \label{eq:ex_disc_test_parameter}
    d = [d^0,d^1],~d \in \mathcal{D} \subseteq \mathcal{X}, \mathrm{and}~\psi = \F_{\infty} \mu \wedge \G_{\infty} \omega,\\
    \llbracket \mu \rrbracket = \left\{g\right\},~\llbracket \omega \rrbracket = \left\{x\in\mathcal{X}~\big|~ x \neq d \right\}.
\end{gather}

\begin{figure}[t]
    \centering
    \includegraphics[width = 0.49 \textwidth]{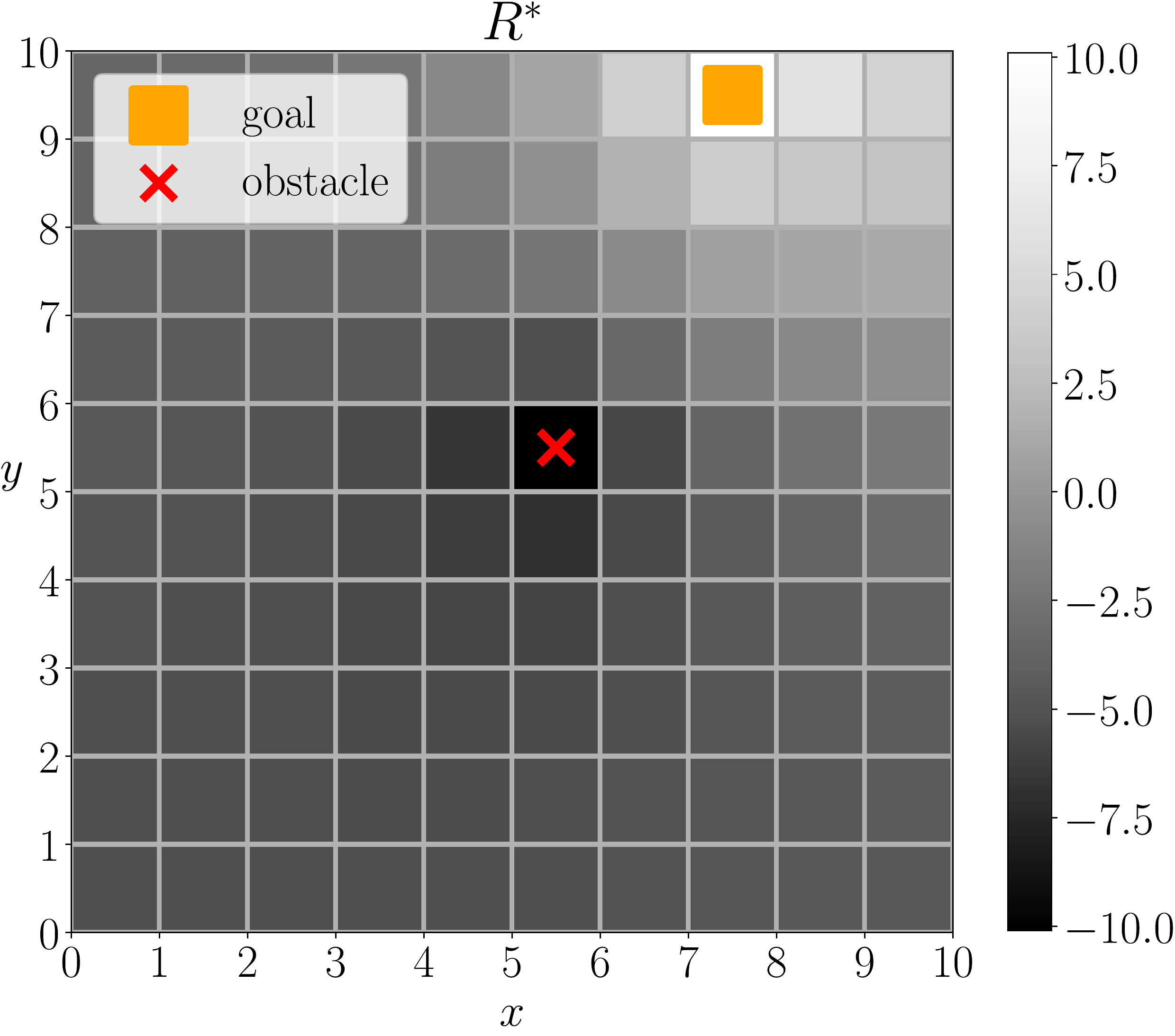}
    \caption{Depiction of $R^*$ utilized to generate example control barrier functions in Section~\ref{sec:disc_examples}, with an example obstacle placed at $o = [5,5]$ for depiction purposes.  $R^*$ is defined in equation~\eqref{eq:discrete_example_cbf}.}
    \label{fig:R^*}
    \vspace{-0.2 in}
\end{figure}

\begin{figure*}[t]
    \centering
    \includegraphics[width = 0.99\textwidth]{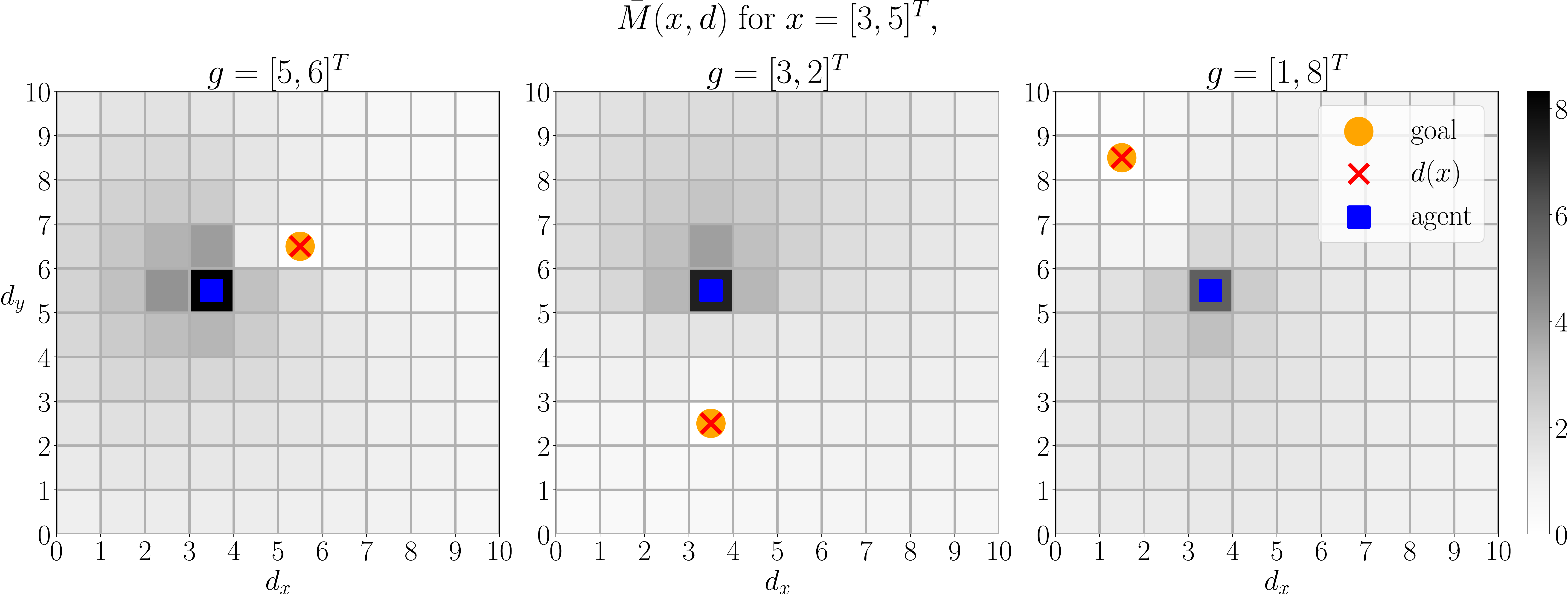}
    \caption{Shown above is a series of color plots indicating the difficulty measure $\bar M$ defined in equation~\eqref{eq:ex_disc_difficulty_measure} for the same system state $x = [3,5]$ and varying goal locations as shown.  The difficulty measure varies with goal location as it depends on $h^F$ which is defined in equation~\eqref{eq:discrete_example_cbf} and varies with the goal.  As can be seen in each of the three cases above however, the test synthesizer~\eqref{eq:ex_disc_feedback_law} accurately identifies a test that minimizes this difficulty measure - it places the obstacle shown via the red "x" over the goal.  The associated colorbar is to the right.}
    \label{fig:discrete_difficulty_measures}
    \vspace{-0.2 in}
\end{figure*}

To construct a variable discrete barrier function then, we will first define a reward matrix $R(d)$ inspired by the work done in~\cite{ahmadi2019safe}.
\begin{align}
    \label{eq:base_reward}
    R(d) = &  \argmax_{V \in \mathbb{R}^{10 \times 10}} & & \hspace{0 in}\mathbf{1}^TV \mathbf{1} \\
    &~\mathrm{subject~to} & & \hspace{0 in} V[g^0,g^1] = 10, \\
    & & & \hspace{0 in} V[d^0,d^1] = -10, \\
    & & & \hspace{0 in} V[x_k^0,x_k^1] = \sum_{u \in \mathcal{U}}~0.2 V[x_{k+1}^0, x_{k+1}^1], \\
    & & & \hspace{0 in} \forall~x_k \in \mathcal{X},~x_k=[x^0_k,x^1_k].
\end{align}
\noindent We note that the optimization problem in equation~\eqref{eq:base_reward} is almost always solvable. For more curious readers, please reference hitting times and absorption probabilities for Markov Chains in~\cite{norris1998markov}.  The only cases precluding a solution occur when the two sets overlap, \textit{i.e.} the goal overlaps with at least one obstacle, yielding an inconsistent feasible set.  As such, we will modify this reward matrix $R(d)$ to generate our barrier functions, \textit{i.e.} $R^*(d) \in \mathbb{R}^{10 \times 10}$, and
\begin{equation}
    \label{eq:discrete_example_cbf}
    \begin{gathered}
        R^*(d)[i,j] = \begin{cases}
            0 & \mbox{if}~\eqref{eq:base_reward}~\mathrm{is~infeasible}, \\
            10.1 & \mbox{if}~i=g^0,~j=g^1, \\
            -10.1 & \mbox{if}~i=d^0,~j=d^1,\\
            R(d)[i,j] &\mbox{else}.
        \end{cases}
        \\
        h^F(x,d) = R^*(d)[x^0,x^1]-10, \\
        h^G(x,d) = R^*(d)[x^0,x^1]+10.
    \end{gathered}
\end{equation}

Figure~\ref{fig:R^*} depicts our resulting $R^*(d)$ after this procedure, for an example case where $d=[5,5]$ and the goal $g = [7,9]$.  
Then, the specific versions of our feedback law and difficulty measure are as follows, with $d$ representing the grid cell location of our single obstacle:
\begin{align}
        \hspace{-0.075 in}\mathcal{U}(x,d) & = \{u \in \mathcal{U}~|~h^G(f(x,u),d) \geq 0\}, \label{eq:ex_disc_feasible_input} \\
        \hspace{-0.075 in}\bar M(x,d) & = \max_{u \in \mathcal{U}}~\mathcal{F}\left(\Delta h^F(x,u,d), u, \mathcal{U}(x,d), -15\right), \label{eq:ex_disc_difficulty_measure}\\
        \hspace{-0.075 in}\testsynth(x) & = \argmin_{d \in \mathcal{D}}\max_{u \in \mathcal{U}}~\mathcal{F}\left(\Delta h^F(x,u,d), u, \mathcal{U}(x,d), -15\right). \label{eq:ex_disc_feedback_law}
\end{align}

\newidea{Figure Analysis:}
Over $1000$ randomized trials wherein the system's initial state and goal are chosen randomly such that they don't overlap, the test synthesizer~\eqref{eq:ex_disc_feedback_law} satisfactorily identifies a test parameter $d$ for each case such that $d$ minimizes the difficulty measure $\bar M$ in~\eqref{eq:ex_disc_difficulty_measure} at that state $x$.  This should also be expected as the proofs in the discrete-time case mirror their continuous counterparts, and the continuous feedback law also exhibited a similar capacity to identify realizable and maximally difficult tests.  To that end, Figure~\ref{fig:discrete_difficulty_measures} portrays both the optimal obstacle setup overlaid on the difficulty measure contour map for the state/goal pair listed.  As can be seen in each of the three cases shown, the test synthesizer~\eqref{eq:ex_disc_feedback_law} accurately identifies an obstacle location that minimizes the corresponding difficulty measure - the color bar indicating evaluations of the difficulty measure is shown on the right-hand side.  Additionally, we can see that in each of the cases, the optimal obstacle location is to place the obstacle directly on top of the goal.  We expect this behavior to be the most difficult test, as in this scenario, there is nothing the system could ever do to reach its goal while simultaneously satisfying its safety specification.  In the following section, we will revisit this example and constrain against such tests to yield a more useful outcome.

\section{Extensions - Constrained Test Synthesis}
\label{sec:extensions}
In this section, we will extend our prior results on the guaranteed realizability and maximal difficulty of our test-synthesis procedure(s), when the space of feasible tests $\mathcal{D}$ is a function of time and/or the system state.  This lets us constrain against trivial test cases like placing obstacles directly on top of the agent/goal.  

\begin{figure*}[t]
    \centering
    \includegraphics[width = \textwidth]{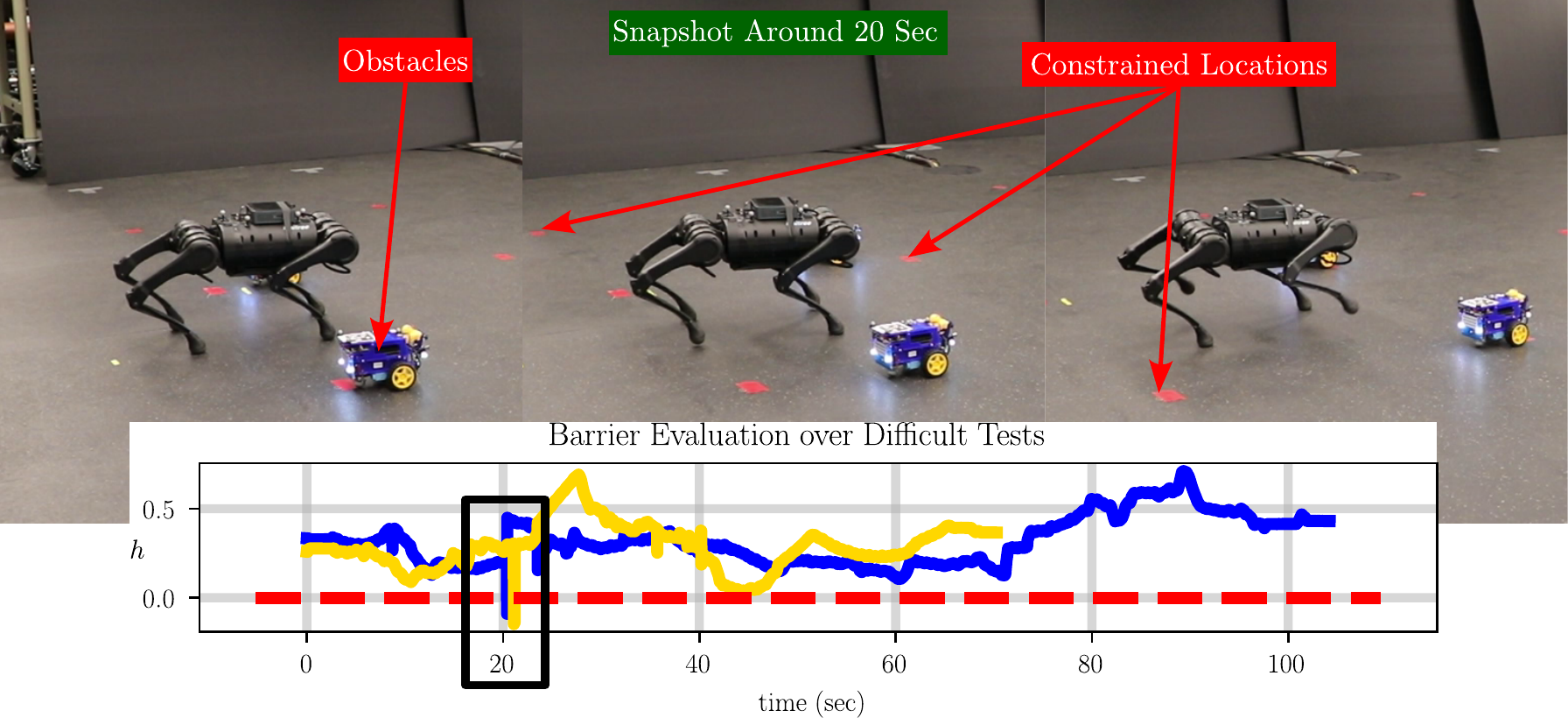}
    \caption{Depiction of the safety failure identified during implementation of the constrained testing procedure described in Section~\ref{sec:constrained_hardware}.  The goal is to determine moving obstacle locations defined at the vertices of a $5 \times 5$ meter grid, while a quadruped ambulates to a goal (off-screen).  Shown above are depictions of the trial whose time-series data for $h$~\eqref{eq:min_global_barrier} is shown in yellow.  Around $20$ seconds, we see the the quadruped (Left) start to try to cut corners between grid cells.  (Middle) This sends a signal to our test-synthesis procedure to ask the obstacle to move to cut off it's path.  (Right) The quadruped reacts to the moving obstacle, but slowly, causing the momentary lapse in safety as signified by the sharp spike in $h$ going negative.  The obstacle waypoints were chosen by solving~\eqref{eq:varying_continuous_feedback_law} with the barrier functions $h^F$ and $h^G_j$ and constrained test map $\mathbb{D}$ provided in equations~\eqref{eq:hardware_barriers} and~\eqref{eq:hardware_testmap} respectively.  Repeating this experiment from the same starting procedure once more, yielded similar behavior as signified by the blue trajectory in the bottom time-series data.}
    \label{fig:Experiment_information}
\end{figure*}

\subsection{Corollaries for Constrained Testing}
To start, we assume there exists a set-valued function that maps from the state space $\mathcal{X}$ and time $t$ to $\mathbb{R}^p$.
\begin{equation}
    \label{eq:test_space_generator}
    \mathbb{D}: \mathcal{X} \times \mathbb{R}_+ \to 2^{\mathbb{R}^p} \suchthat \mathbb{D}(x,t) = \mathcal{D} \subset \mathbb{R}^p.
\end{equation}
This results in the following change to the test synthesizer for the continuous setting in equation~\eqref{eq:feedback_law}:
\begin{gather}
    \label{eq:varying_continuous_feedback_law}
    \testsynth(x,t) = \argmin_{d \in \mathbb{D}(x,t)}\max_{u \in \mathcal{U}}~\mathcal{F}\left(\dot h^F(x,d,u), u, \mathcal{U}(x,d), m\right), \\
    m \leq \min_{u \in \mathcal{U},~x \in \mathcal{X},~d \in \mathbb{D}(x,t)}~\dot{h}^F(x,d,u).
\end{gather}
In the discrete setting, we will directly mention the change with respect to the predictive test-synthesis framework.  Our test-synthesis procedure is as follows, with $\Delta$ as defined in equation~\eqref{eq:predictive_difference}.
\begin{align}
    & \testsynth(x,t) = \argmin_{d \in \mathbb{D}(x,t)}~\max_{\mathbf{u} \in \mathcal{U}^N}~\xi(x,\mathbf{u},d), \label{eq:varying_discrete_feedback_law}\\
    & \xi(x,\mathbf{u},d) = \mathcal{F}\left(\Delta h^F\left(x, \mathbf{u},d\right), \mathbf{u}, \mathcal{U}^N(x,d),m\right), \\
    & m \leq \min_{\mathbf{u} \in \mathcal{U}^N, x \in \mathcal{X}, d \in \mathbb{D}(x,t)}~\Delta h^F\left(x,\mathbf{u},d\right).
\end{align}
As we have allowed the space of feasible tests to vary, we need to change the statements of Assumptions~\ref{assump:continous_assumption} and~\ref{assump:finite_spaces} to match.  The analog of Assumption~\ref{assump:continous_assumption} is as follows.
\begin{assumption}
\label{assump:continuous_constrained}
Each feasible test space $\mathcal{D} \in \mathcal{R}(\mathbb{D})$ is a compact set, the input space $\mathcal{U}$ for the continuous time system~\eqref{eq:nom_sys} is a closed, convex polytope in $\mathbb{R}^m$, and $h^F,h^G_j \in C^1(\mathcal{X} \times \mathcal{D})$.
\end{assumption}
\noindent Likewise, the analog of Assumption~\ref{assump:finite_spaces} is as follows.
\begin{assumption}
\label{assump:discrete_cosntrained}
The input space $\mathcal{U}$ for the discrete-time system~\eqref{eq:discrete_system} is finite, and each feasible test space $\mathcal{D} \in \mathcal{R}(\mathbb{D})$ is also finite.
\end{assumption}
With these two new feedback laws and assumptions, we can state and prove the following corollaries guaranteeing the realizability of the tests generated by each feedback law.
\begin{corollary}
\label{cor:existence_varying_continuous}
Let Assumption~\ref{assump:continuous_constrained} hold.  The test synthesizer in~\eqref{eq:varying_continuous_feedback_law} is guaranteed to have a solution $d \in \mathbb{D}(x,t)~\forall~x \in \mathcal{X},~t\in\mathbb{R}_+$, \textit{i.e.},
\begin{equation}
    \forall~x \in \mathcal{X},~t \in \mathbb{R}_+~\exists~d \in \mathbb{D}(x,t) \suchthat d = \testsynth(x,t).
\end{equation}
\end{corollary}
\begin{proof}
This is an application of Theorem~\ref{thm:existence} for each $\mathcal{D} \in \mathcal{R}(\mathbb{D})$ which is assumed to be compact via Assumption~\ref{assump:continuous_constrained}.
\end{proof}

\begin{corollary}
\label{cor:disc_varying_existence}
Let Assumption~\ref{assump:discrete_cosntrained} hold.  The test synthesizer in~\eqref{eq:varying_discrete_feedback_law} is guaranteed to have a solution $d \in \mathbb{D}(x,t)~\forall~x \in \mathcal{X},~t\in\mathbb{R}_+$, \textit{i.e.},
\begin{equation}
    \forall~x \in \mathcal{X},~t \in \mathbb{R}_+~\exists~d \in \mathbb{D}(x,t) \suchthat d = \testsynth(x,t).
\end{equation}
\end{corollary}
\begin{proof}
This is an application of Corollary~\ref{cor:disc_predictive_existence} for each $\mathcal{D} \in \mathcal{R}(\mathbb{D})$ which is assumed to be finite via Assumption~\ref{assump:discrete_cosntrained}.
\end{proof}

These modified test-generation laws also minimize their respective difficulty measures over the constrained test-generation set $\mathbb{D}(x,t)$.

\begin{corollary}
\label{cor:cont_varying_optimality}
Let Assumption~\ref{assump:continuous_constrained} hold. The test synthesizer~\eqref{eq:varying_continuous_feedback_law} minimizes the difficulty measure $M$~\eqref{eq:difficulty_measure} over all $d \in \mathbb{D}(x,t)$, \textit{i.e.}
\begin{equation}
    \testsynth(x) = \argmin_{d \in \mathbb{D}(x,t)}~M(x,d).
\end{equation}
\end{corollary}
\begin{proof}
This proof stems via application of Corollary~\ref{corr:optimality} for each compact set $\mathcal{D} \in \mathcal{R}(\mathbb{D})$.
\end{proof}

In the discrete setting we have the following corollary.
\begin{corollary}
\label{cor:disc_varying_optimality}
Let Assumption~\ref{assump:discrete_cosntrained} hold.  The test synthesizer~\eqref{eq:varying_discrete_feedback_law} minimizes the difficulty measure $\Tilde M$~\eqref{eq:discrete_difficulty_measure_predictive} over all $d \in \mathbb{D}(x,t)$, \textit{i.e.}
\begin{equation}
    \testsynth(x) = \argmin_{d \in \mathbb{D}(x,t)}~\Tilde{M}(x,d).
\end{equation}
\end{corollary}
\begin{proof}
This proof stems via repeated application of Corollary~\ref{cor:disc_predictive_optimality} for each finite set $\mathcal{D} \in \mathcal{R}(\mathbb{D})$.
\end{proof}

\newidea{Remark on Environment Dynamics:}  Sometimes, there may exist torque bounds on $\dot d$, \textit{e.g.} environment dynamic constraints, and as stated the mentioned approach could not account for such constraints.  However, one could augment the system state to include the nominal, modelable state of the environment, \textit{e.g.} the overall ``state" $x$ would be the state of the system-under-test $x_s$ and the state of the modelable aspect of the environment subject to dynamic constraints $d_t$.  Then, the test-generation procedure would provide the dynamics $d$ that the modelable environment state $d_t$ must follow.  Study in this vein is the subject of current work by the authors and collaborators, see~\cite{badithela2021synthesis}.

\subsection{Applications to a Constrained Hardware Test}
\label{sec:constrained_hardware}
The constrained test-synthesis results permit us to start applying our procedure to testing hardware systems in their operating environments.  Specifically, we test a quadruped's ability to navigate to a goal while avoiding moving robots attempting to block its path.  To start, we idealize the quadruped as a single integrator system:
\begin{equation}
    \label{eq:hardware_system}
    \dot x = u,~x \in \mathcal{X} \triangleq [-1,4] \times [-2,3],~u \in \mathcal{U} \triangleq [-5,5]^2.
\end{equation}

To construct our specification $\psi$ and our test parameter vector $d$, we will denote the goal as $g = [3.5,2.5]^T$ and the obstacle agent locations on the 2-d plane as $o_j \in \mathcal{X}$.  Then $\psi$ is as follows:
\begin{gather}
    \llbracket \mu \rrbracket = \{x \in \mathbb{R}^2~|~\|x-g\| \leq 0.3 \}, \\ \llbracket \omega_j \rrbracket = \{x \in \mathbb{R}^2~|~\|x - o_j\| \geq 0.3\}, \\
    \psi = \F_{\infty}\mu \wedge_{j=1,2} \G_{\infty}\omega_j.   \label{eq:hardware_spec}
\end{gather}
The combined locations of these obstacles will be our test parameter vector, \textit{i.e.} $d = [o_1^T, o_2^T]^T \in \mathcal{X}^2$.

To generate tests, 
our barrier functions $h^F$ and $h^G_j$ and test space map $\mathbb{D}$ are as follows (here $x$ is the quadruped's planar position, \textit{i.e.} $x = [x_1,x_2] \in \mathcal{X}$ and for a scalar $\epsilon \in \mathbb{R}$, $\lfloor \epsilon \rfloor$ denotes rounding down and $\lceil \epsilon \rceil$ denotes rounding up):
\begin{gather}
    h^F(x) = 0.3 - \|x - g\|,~h^G_j(x,d) = \|x-o_j\| - 0.3, \label{eq:hardware_barriers} \\
    d = [o_1^T,o_2^T]^T \in \mathbb{D}(x) \triangleq \{\lfloor x_1 \rfloor, \lfloor x_2 \rfloor \} \times \{\lceil x_1 \rceil, \lceil x_2 \rceil \}.~\label{eq:hardware_testmap}
\end{gather}
As an example then $\mathbb{D}(x = [0.3,1.7]) = \{0,1\} \times \{1, 2\}$.

For testing purposes then, we will calculate the optimal obstacle locations offline and direct the obstacles to move between the identified grid points at test time.  Theoretically, the identified grid points should correspond to the most difficult test of quadruped behavior while it ambulates within a grid cell - asking the obstacles to move between grid points as the quadruped moves between grid cells should likewise be more difficult.  We repeated the experiment twice, and recorded the minimum value of both barrier functions:
\begin{equation}
    \label{eq:min_global_barrier}
    h(x,d) = \min_{j=1,2}~h^G_j(x,d),
\end{equation}
over the course of the entire multi-system trajectory.  The corresponding time-series data is shown in Figure~\ref{fig:Experiment_information}.  

In both tests, the quadruped tries "cutting" corners between grid cells, as at least one of the obstacles moves to the associated grid point that the quadruped is trying to move through.  This causes a momentary loss of safety as evidenced by the sharp spike where the minimum barrier value goes negative (around $20$ seconds).  The quadruped quickly corrects this mistake and resumes its normal trajectory while repeating this cutting behavior a few more times.  Although, it maintains a positive barrier value both times - this references the "dip" in the yellow trajectory around $45$ seconds, and the two "dips" in the blue trajectory around $60$ and $70$ seconds.  Figure~\ref{fig:Experiment_information} depicts this momentary loss of safety.  This procedure, however, shows an example implementation of our worst-case tests on hardware systems.  Additionally, it shows that the tests that we theorize to be the most difficult uncover problematic system behavior as required of our testing procedure.

\section{Conclusion}
In this paper, we presented an adversarial approach to test synthesis for autonomous systems based on control barrier functions and timed reach-avoid specifications.  We prove that our approach will always produce realizable and maximally difficult tests of system behavior as our synthesis techniques are guaranteed to have solutions that minimize a corresponding difficulty measure - a concept we introduce and define.  Finally, we show the efficacy of our procedure in generating tests for simple toy examples useful in a real-world context - unicycle systems and grid-world abstractions, both of which are used for baseline navigation control algorithms in other works.  We also show how such an abstraction can easily be extended to a useful hardware system test - testing a quadruped's ability to navigate within a grid while avoiding obstacles.

\section{Acknowledgements}
We would like to thank Ryan Cosner and Wyatt Ubellacker for their tremendous help in running experiments.  Additionally, we would like to thank Apurva Badithela and Josefine Graebner for their thought provoking discussions regarding problem formulation and potential solutions.  Finally, Prithvi Akella was also supported by the Air Force Office of Scientific Research, grant FA9550-19-1-0302.

\bibliographystyle{ieeetr}
\bibliography{bib_works}

\newpage

\begin{IEEEbiography}[{\includegraphics[width = 1 in, height = 1.25 in, clip,keepaspectratio]{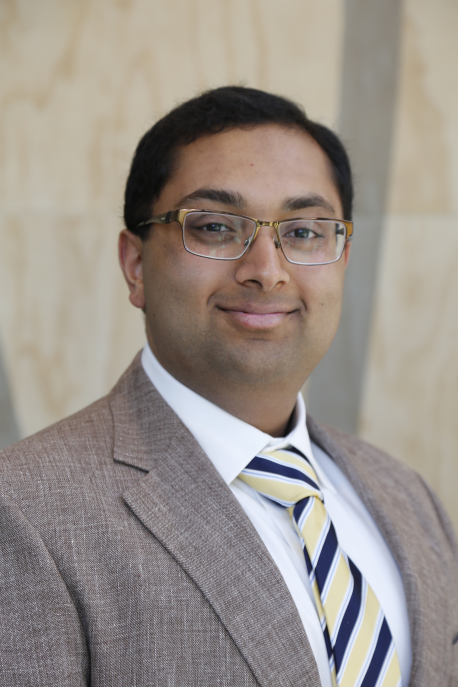}}]{Prithvi Akella}
received the B.S. degree in Mechanical Engineering from the University of California, Berkeley, in 2018 and the M.S. degree in Mechanical Engineering from California Institute of Technology in 2020.  He is the recipient of the Bell Family Graduate Fellowship in Engineering and Applied Sciences.  His current research focuses on the automated test and evaluation of cyber-physical systems.
\end{IEEEbiography}

\begin{IEEEbiography}[{\includegraphics[width=1in,height=1.25in,clip,keepaspectratio]{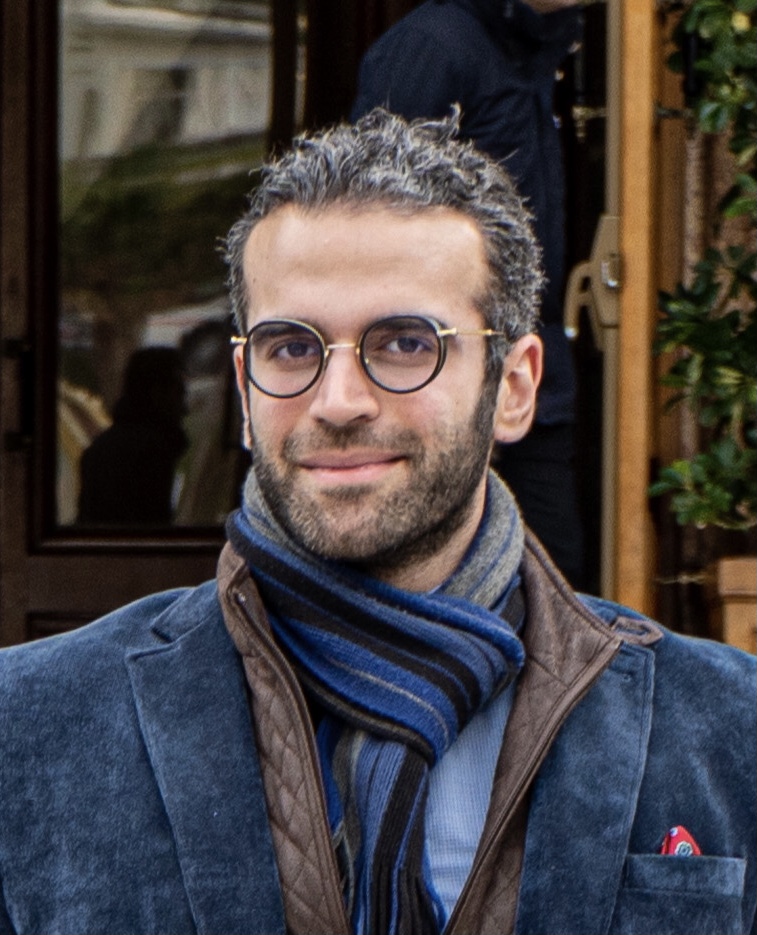}}]{Mohamadreza Ahmadi}
 is a principal research scientist in the planning and control group at TuSimple. He finished his DPhil in Engineering Science--Control Systems and Aeronautics--in November 2016 at the University of Oxford, UK, as a Clarendon Scholar. His PhD was followed by postdoctoral studies at the center for autonomous systems and technologies (CAST) at the California Institute of technology. He is the recipient of the Sloan-Robinson Engineering Fellowship, an Edgell-Sheppee Award, and an ICES Postdoctoral Fellowship. His current research is on planning and control under uncertainty with application to autonomous vehicles, in particular, self-driving trucks. 
\end{IEEEbiography}

\begin{IEEEbiography}[{\includegraphics[width=1in,height=1.25in,clip,keepaspectratio]{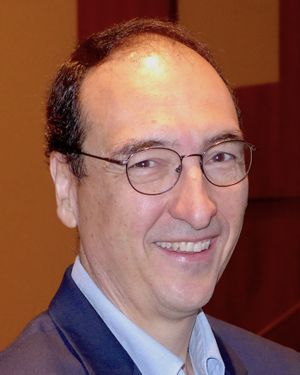}}]{Richard M. Murray}
 received the B.S.
degree in electrical engineering from California Institute of Technology (Caltech),
Pasadena, CA, USA, in 1985 and the M.S.
and Ph.D. degrees in electrical engineering
and computer sciences from the University
of California at Berkeley, Berkeley, CA, USA,
in 1988 and 1991, respectively.
He is currently the Thomas E. and Doris
Everhart Professor of Control \& Dynamical Systems and Bioengineering at Caltech. His research is in the application of feedback
and control to networked systems, with applications in biology
and autonomy. Current projects include analysis and design of biomolecular feedback circuits, synthesis of discrete decision-making
protocols for reactive systems, and design of highly resilient architectures for autonomous systems
\end{IEEEbiography}

\begin{IEEEbiography}[{\includegraphics[width=1in,height=1.25in,clip,keepaspectratio]{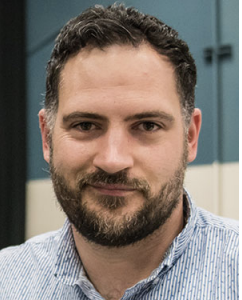}}]{Aaron D. Ames}
 received the B.S. degree in mechanical engineering and the B.A. degree in mathematics from
the University of St. Thomas, Saint Paul, MN,
USA, in 2001, the M.A. degree in mathematics
from the University of California at Berkeley, in
2006, and the Ph.D. degree in EECS from UC
Berkeley, Berkeley, CA, USA, in 2006.
He is currently the Bren Professor of mechanical and civil engineering and control and
dynamical systems with Caltech, Pasadena,
CA, USA. His research interests span the areas of robotics, nonlinear, safety-critical control, and hybrid systems, with a special focus
on applications to bipedal robotic walking both formally and through
experimental validation.
\end{IEEEbiography}

\end{document}